\documentclass[a4paper,UKenglish,cleveref, autoref, thm-restate]{lipics-v2021}
\usepackage[T1]{fontenc}

\usepackage[utf8]{inputenc}
\usepackage{microtype}
\usepackage{stmaryrd}
\usepackage{extarrows}
\usepackage{comment}
\usepackage{thm-restate,thmtools}

\usepackage{tikz}
\usetikzlibrary{positioning,shapes,petri,matrix}
\usetikzlibrary{automata}
\usetikzlibrary{arrows,backgrounds,positioning,fit,calc}

\usepackage{tikzsymbols}
\usepackage{xspace} 
\usepackage{algorithm,algpseudocode,booktabs}
\usepackage{pgfplots}

\usepackage{wrapfig}
\usepackage{graphicx}
\usepackage{array,multirow}
\usepackage{url}
\usepackage{xcolor}
\usepackage{amsmath, amsfonts, amssymb}
\usepackage{scalerel}

\pdfoutput=1 
\hideLIPIcs  

\author{Anirban Majumdar}{Independent Researcher \and \url{https://anirban11.github.io}}{majumdaranirban963@gmail.com}{https://orcid.org/0000-0003-4793-1892}{}

\author{Sayan Mukherjee}{Univ Rennes, Inria, CNRS, IRISA, France \and \url{https://mukherjee-sayan.github.io}}{sayan.mukherjee@irisa.fr}{https://orcid.org/0000-0001-6473-3172}{}

\author{Jean-François Raskin}{Université Libre de Bruxelles, Belgium \and \url{https://verif.ulb.ac.be/jfr/}}{jean-francois.raskin@ulb.be}{https://orcid.org/0000-0002-3673-1097}{}

\authorrunning{Anirban Majumdar, Sayan Mukherjee and Jean-François Raskin} 



\supplement{The source codes related to our algorithm can be found as follows:}
\supplementdetails[subcategory={Source Code}]{Software}{https://github.com/anirban11/leap}
\supplementdetails[subcategory={}]{Artifact for reproducing the experiments}{https://zenodo.org/records/15847543}

\acknowledgements{Most of this work was done while Anirban Majumdar and Sayan Mukherjee were post-doctoral researchers at Université Libre de Bruxelles, Belgium.}

\nolinenumbers 

\title{Learning Event-recording Automata Passively}

\usepackage{macros}

\begin{document}
\maketitle

\begin{abstract}
This paper presents a state-merging algorithm for learning timed languages definable by Event-Recording Automata (ERA) using positive and negative samples in the form of symbolic timed words.  Our algorithm, $\leap$ (Learning Event-recording Automata Passively), constructs a possibly nondeterministic ERA from such samples based on merging techniques. We prove that determining whether two ERA states can be merged while preserving sample consistency is an $\np$-complete problem, and address this with a practical SMT-based solution. Our implementation demonstrates the algorithm's effectiveness through examples. We also show that every ERA-definable language can be inferred using our algorithm with a suitable sample.
\end{abstract}


\section{Introduction}
\label{sec:intro}

Formal modeling is essential in fields like requirements engineering and computer-aided verification, where it supports the analysis of complex systems. However, creating formal models is labor-intensive and prone to error. Learning algorithms can streamline this process by deriving models from execution scenario descriptions~\cite{DBLP:journals/cacm/Vaandrager17}. When the set of possible behaviors of a system aligns with a regular language, it can be represented by a finite state automaton, and passive and active learning algorithms and efficient implementations are available to derive the minimal Deterministic Finite Automaton (DFA) describing this language. In contrast, learning timed languages remains underdeveloped.
In this paper, we contribute to this line of research by providing a passive learning algorithm, and an implementation, that transforms sets of positive and negative \emph{symbolic timed words} into \emph{Event-Recording Automata} (ERA, for short)~\cite{AFH99}. 

The class of Event-Recording Automata (ERA) is a subclass of Timed Automata (TA)~\cite{AD94}, where each clock $x_{\sigma}$ is linked to an event $\sigma \in \Sigma$ and records the time since its last occurrence. ERA offer key advantages over TA: they are determinizable and complementable, making them suitable for specification in verification tasks, as inclusion checking is decidable for ERA but not for TA. Additionally, clocks in ERA are directly tied to events and reset automatically with each event occurrence, enhancing {\em interpretability} compared to TA and making ERA appealing for learning applications where trust is important. Further, it is a folklore result that ERA combined with homomorphism is equally expressive as TA. This means, every system that can be modelled as a timed automaton can also be modelled as an ERA, by increasing the alphabet. 

Before detailing our contributions, let us consider an example that illustrates how formal specifications, in the form of ERA, can be learned from {\em symbolic} scenarios, as advocated in~\cite{holbrook1990scenario}, and expressed here as {\em symbolic timed words}.

\medskip
\noindent\textbf{Learning ERA from symbolic timed words.}
Let us consider the task for a \emph{requirement engineer} (RE) to write a formal model in the form of an ERA, for the following informal requirement: {\it whenever a (button) \press\ is followed by another (button) \press\ within $1$ time unit, then an \alarm\ must happen immediately.}

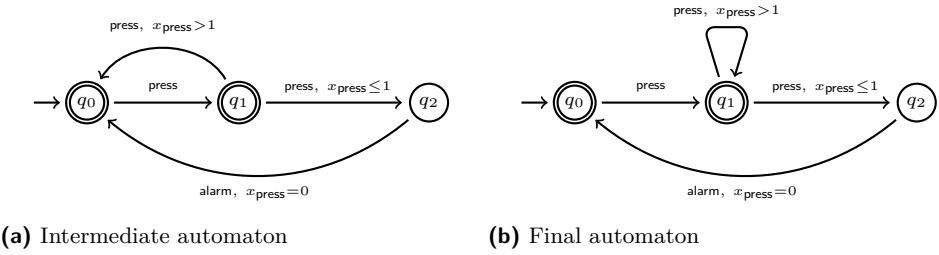
\begin{figure}[t]
    \centering
    \begin{subfigure}{.45\textwidth}
    \centering
        \begin{tikzpicture}
        \everymath{\scriptstyle}
        \begin{scope}[every node/.style={circle, draw, inner sep=2pt,
            minimum size = 3mm, outer sep=3pt, thick}] 
            \node [double] (0) at (0,0) {$q_0$}; 
            \node [double] (1) at (2,0) {$q_1$};
            \node [] (2) at (4.5,0) {$q_2$};
        \end{scope}
        \begin{scope}[->, thick]
            \draw (-0.7,0) to (0); 
            \draw (0) to (1);
            \draw (1) to (2);
            \draw [bend right = 60] (1) to (0);
            \draw [bend left = 40] (2) to (0);
        \end{scope}
        \node at (1,0.2) {\scriptsize $\press$};
        \node at (1,1) {\scriptsize $\press,~x_\press > 1$};
        \node at (3.3,0.2) {\scriptsize $\press,~x_\press \leq 1$};
        \node at (2.2,-1.2) {\scriptsize $\alarm,~x_\press = 0$};
    \end{tikzpicture}
    \caption{Intermediate automaton}
    \label{fig:motivation-1}
    \end{subfigure}
    \begin{subfigure}{.45\textwidth}
    \centering
        \begin{tikzpicture}
        \everymath{\scriptstyle}
        \begin{scope}[every node/.style={circle, draw, inner sep=2pt,
            minimum size = 3mm, outer sep=3pt, thick}] 
            \node [double] (0) at (0,0) {$q_0$}; 
            \node [double] (1) at (2,0) {$q_1$};
            \node [] (2) at (4.5,0) {$q_2$};
        \end{scope}
        \begin{scope}[->, thick]
            \draw (-0.7,0) to (0); 
            \draw (0) to (1);
            \draw (1) to (2);
            \draw [rounded corners] (1) to (1.7,1) to (2.3, 1) to (1);
            \draw [bend left = 40] (2) to (0);
        \end{scope}
        \node at (1,0.2) {\scriptsize $\press$};
        \node at (2,1.2) {\scriptsize $\press,~x_\press > 1$};
        \node at (3.3,0.2) {\scriptsize $\press,~x_\press \leq 1$};
        \node at (2.2,-1.2) {\scriptsize $\alarm,~x_\press = 0$};
    \end{tikzpicture}
    \caption{Final automaton}
    \label{fig:motivation-2}
    \end{subfigure}
    \caption{Motivation for choosing symbolic words as specification}
    \label{fig:motivation}
\end{figure}

We demonstrate how the RE can use our learning algorithm to obtain the corresponding ERA. The RE must provide positive and negative examples in the form of scenarios. Instead of requiring explicit timed words, i.e., $(\sigma_1,t_1)(\sigma_2,t_2)\dots(\sigma_n,t_n)$, which are sequences of events (e.g., {\sf press} or {\sf alarm}) with timestamps $t_i \in \mathbb{R}_{\geq 0}$, we allow more abstract examples in the form of \emph{symbolic timed words}. These are sequences of the form $(\sigma_1,g_1)(\sigma_2,g_2)\dots(\sigma_n,g_n)$, where $\sigma_i$ are events and $g_i$ are conjunctions of constraints of the form $x_{\sigma_i} \in I$. Here, $x_{\sigma_i}$ is the event-recording clock associated with $\sigma_i$, and $I$ is an interval with non-negative integer bounds. Symbolic timed words provide higher-level information, which is natural for requirements elicitation.

Consider a positive example: a \press\ occurs, followed by another \press\ within 1 time unit, and then an \alarm\ happens immediately. This translates into the symbolic timed word $(\press,\top)(\press, x_{\press} \leq 1)(\alarm, x_{\press} = 0)$.
A second example is: if a \press\ occurs followed by another \press\ after 1 time unit, and nothing happens afterwards, it is acceptable. This is represented as $(\press, \top)(\press, x_{\press} > 1)$. Additionally, trivial scenarios include: if nothing happens, it is acceptable (i.e., $\varepsilon$ is a positive example), and if a \press\ occurs followed by nothing, that is also fine (i.e., $(\press, \top)$ is a positive example).

The RE must also provide some {\em undesired} scenarios as negative samples to avoid the learning algorithm producing a one-state universal automaton that allows all scenarios. A first negative example could be: a \press\ followed by another \press\ after 1 time unit and then an \alarm\ happens. This is specified as $(\press,\top)(\press, x_{\press} > 1)(\alarm, x_{\press} \geq 0)$. Another negative example is a \press\ followed by another \press\ within 1 time unit and no alarm afterward, formalized by $(\press, \top)(\press, x_{\press} \leq 1)$. Also, an \alarm\ should not occur immediately after the first \press, which can be specified by $(\press, \top)(\alarm, x_{\press} \geq 0)$.

Using these samples, our algorithm will generate the ERA in Figure~\ref{fig:motivation-1}. By inspecting this ERA, the RE may discover that the model allows a scenario where a \press\ is followed by another \press\ after $1$ time unit, then followed by a \press\ immediately, which is not desired. This can be provided as a negative word $(\press, \top)(\press, x_\press > 1)(\press, x_\press = 0)$. The algorithm will then compute the ERA in Figure~\ref{fig:motivation-2}, which correctly models the target specification.

This example demonstrates that using symbolic timed words, and not plain timed words, is natural in the process of {\em requirements elicitation} and that our learning algorithm effectively helps the RE generalize natural scenarios that are easy to formulate.

\medskip
\noindent\textbf{Context and contributions.} 
This paper focuses on passive learning from positive and negative examples, a concept pioneered by Gold~\cite{DBLP:journals/iandc/Gold67,DBLP:journals/iandc/Gold78}, who defined conditions under which a language from a given class can be identified based on sample data. Building on Gold's insights, the {\em Regular Positive and Negative Inference} (RPNI) algorithm~\cite{oncina1992} emerged as a key technique for identifying regular languages from examples, using state-merging techniques. RPNI starts with positive and negative examples, treating each positive prefix as a distinct state, then progressively merges states that are compatible with the samples. For regular languages, this approach ensures the result is a DFA consistent with the examples and converges to the correct language given enough data. Passive learning has applications in synthesis, such as serving as a subroutine in frameworks like SyGuS~\cite{SyGus} or allowing users to provide guiding examples to simplify synthesis outputs, as suggested in~\cite{DBLP:conf/tacas/BalachanderFR23}.
\smallskip

Our main technical contributions are threefold. First, we propose an algorithm to infer Event-Recording Automata (ERA) compatible with specific sets of positive and negative symbolic timed words (Section~\ref{sec:algo}), providing a new algorithm for passive learning for timed systems from high level symbolic examples, and not from plain timed words. However, this requires addressing a computational challenge: 
as a second contribution, we establish that determining whether a symbolic timed word intersects the language of an ERA, even if deterministic, is an $\np$-complete problem (Section~\ref{sec:complexity}). This problem is essential to our state-merging algorithm for deciding permissible state merges, and we address this complexity issue using a reduction to SMT. Our third contribution proves that {\leap} is \emph{complete}: for every ERA-recognizable language, there exist finite characteristic sets of positive and negative examples such that our algorithm returns an ERA accurately representing the language (Section~\ref{sec:leap-completeness-sdera}). We achieve this via a novel language-theoretic approach, showing that the final automaton’s language aligns with the target language, a method potentially extendable to other models without unique minimal automata. We have implemented {\leap} in a \textsc{Python} prototype, demonstrating its effectiveness on various examples (Section~\ref{sec:experiments}). Definitions and notations appear in Section~\ref{sec:prelims}.

\medskip
\noindent\textbf{Related works.} 
Several methods have been proposed for inferring ERA in different frameworks. Works such as~\cite{GJL10,GJP06,LADSL11} introduce {\em active} learning algorithms for ERA-recognizable languages, while a recent contribution~\cite{ATVA2024} adapts a separability-based algorithm to infer ERA with a minimal number of control states. Other efforts target more expressive models: for example, \cite{HJM20} infers automata from an extended ERA class, while \cite{Waga23,DTA-hscc} learn the full class of Deterministic Timed Automata. Additional algorithms target subclasses of automata, such as One-clock Timed Automata~\cite{octa,VWW09}, Real-time Automata~\cite{AWZZZ21}, and Mealy Machines with Timers~\cite{VEB23,BGPSV24}. 
Being in the active learning framework, all these algorithms assume access to a teacher, who answers queries that progressively refine the target automaton. For instance, algorithms in~\cite{GJP06,HJM20} use state-merging with tree-based data structures that represent membership queries. They apply restricted merging criteria requiring that the tree is complete up to a certain depth to cover the target language. Also, all those techniques maintain trees and automata that are deterministic. In contrast, our approach addresses a fixed sample of words and allows more flexible merging, yielding potentially smaller non-deterministic automata (more details in Section~\ref{subsec:comparison-state-merging}).  Similarly, $L^{\#}$~\cite{Lsharp}, an active learning algorithm for regular languages, merges states based on apartness but relies on sufficient data about the target language, making it not directly applicable to passive learning. 

In the {\em passive} learning framework, fewer works exist~\cite{VWW10,PMM17,CLRT22}, where subclasses of TA are inferred from {\em timed words}.
Also, the models that are considered are different. Of these, \cite{VWW10,CLRT22} target Deterministic Real-time Automata, a weaker model than ERA, while \cite{PMM17,CLRT22} only consider positive data, with~\cite{VWW10} incorporating both positive and negative samples. These works model systems from logs, typically in the form of timed words, with \cite{TALL19} proposing a Genetic programming approach for TA.
An SMT-based algorithm for inferring TA has also been studied in~\cite{TAL22}.
While timed words are well-suited for learning from execution logs and extracting models from software treated as a black box, they are less appealing for our intended use case. Here, a requirements engineer (RE) aims to craft an ERA based on scenarios. Our approach uses symbolic words (positive and negative) as input, enabling the RE to provide a compact set of scenarios and construct smaller automata (see Section~\ref{sec:experiments} for empirical evidence of this).


\section{Preliminaries}
\label{sec:prelims}

\noindent\textbf{Timed words and timed languages.}
A \emph{timed word} over an alphabet $\Sigma$ is a finite sequence $(\sigma_1,t_1)(\sigma_2,t_2)\dots (\sigma_n,t_n)$ where each $\sigma_i \in \Sigma$ and $t_i \in \mathbb{R}_{\geq 0}$, for all $1 \leq i \leq n$, and for all $1 \leq i < j \leq n$, $t_i \leq t_j$ (time is monotonically increasing). We use $\mathsf{TW}(\Sigma)$ to denote the set of all timed words over the alphabet $\Sigma$.

A \emph{timed language} is a (possibly infinite) set of timed words.
Timed Automata (TA)~\cite{AD94} extend deterministic finite-state automata with \emph{clocks}. In what follows, we will 
use a subclass of TA, where clocks have a pre-assigned semantics and are not reset arbitrarily. This class is known as the class of Event-recording Automata (ERA)~\cite{AFH99}.
We now define the necessary vocabulary and notations for their definition. 
\medskip

\noindent\textbf{Constraints.} A \emph{clock} is a non-negative real valued variable, that is, a variable ranging over $\mathbb{R}_{\geq 0}$.
Let $K$ be a positive integer.
A \emph{$K$-atomic constraint} over a clock $x$, is an expression of the form $x = c$, $x \in (c,d)$ or $x > K$ where $c,d \in \mathbb{N} \cap [0,K]$.
We also consider a \emph{trivial} $K$-atomic constraint $\top$. 
A \emph{$K$-constraint} over a set of clocks $X$ is a conjunction of $K$-atomic constraints over clocks in $X$.

A \emph{$K$-elementary constraint} over a clock $x$, is an atomic constraint where the interval is restricted to unit intervals; more formally, it is an expression of the form $x = c$, $x \in (c,c+1)$ or $x > K$ where $c,c+1 \in \mathbb{N} \cap [0,K]$. 
A \emph{$K$-simple constraint} over $X$ is a conjunction of $K$-elementary constraints over clocks in $X$, 
where each variable $x \in X$ appears exactly in one conjunct. 
The definition of simple constraints also appear in earlier works~\cite{GJL10}.

For example,
    let $X = \{x_1, x_2, x_3\}$ be a set of clocks,
    and let $K = 2$.
     Then, $0 < x_1 < 2$ (that is, $x_1 \in (0,2)$) is a $K$-atomic constraint, but not a $K$-elementary constraint since the interval is not a unit interval; 
    $1 < x_1 < 2$ (that is, $x_1 \in (1,2)$) is an example of a $K$-elementary constraint;
     $x_1 > 2 \wedge 1 < x_3 < 2$ is a $K$-constraint, but not a $K$-simple constraint since the clock $x_2$ does not appear in this constraint; 
    $x_1 > 2 \wedge 1 < x_3 < 2 \wedge x_2 = 2$ is a $K$-simple constraint.
\medskip

\noindent\textbf{Valuation and satisfaction of constraints.}
A \emph{valuation} for a set of clocks $X$ is a function ${v \colon X \rightarrow \mathbb{R}_{\geq 0}}$.
A valuation $v$ for $X$ satisfies a $K$-atomic constraint $\psi$, written as $v \models \psi$, if: when every clock $x$ appearing in $\psi$ is replaced with its value $v(x)$, the expression evaluates to True.
Every valuation satisfies $\top$.

A valuation $v$ satisfies a $K$-constraint over $X$ if $v$ satisfies all its conjuncts.
Given a $K$-constraint $\psi$, $\sem{\psi}$ will denote the set of all valuations $v$ such that $v \models \psi$.
It is easy to verify that for any $K$-simple constraint $r$ and any $K$-elementary constraint $\psi$, either $\sem{r} \cap \sem{\psi}=\emptyset$ or $\sem{r} \subseteq \sem{\psi}$.
W.l.o.g., in this work, we will only consider $K$-constraints $g$ s.t. $\sem{g} \neq \emptyset$.

\begin{lemma}
    \label{lem:rws-dont-intersect}
    Let $K$ be a positive integer. Then
    for any two $K$-simple constraints $r_1$ and $r_2$, either $\sem{r_1} = \sem{r_2}$ or $\sem{r_1} \cap \sem{r_2} = \emptyset$.
\end{lemma}

Let $\Sigma=\{\sigma_1,\sigma_2,\dots,\sigma_k\}$ be a finite \emph{alphabet}. The set of all \emph{event-recording clocks} associated with $\Sigma$ is denoted by $X_{\Sigma}=\{ x_{\sigma} \mid \sigma \in \Sigma\}$. We denote by $\SC{\Sigma, K}$ the set of all $K$-constraints over the set of clocks $X_{\Sigma}$ and we use $\SimC{\Sigma, K}$ to denote the set of all $K$-simple constraints over the clocks in $X_{\Sigma}$. 
Since $K$-simple constraints are also $K$-constraints, we have that $\SimC{\Sigma, K} \subset \SC{\Sigma, K}$.
We will omit $K$ when not important, or if it is clear from context, and denote the above sets by $\SimC{\Sigma}$ and $\SC{\Sigma}$, respectively.
\medskip

\noindent\textbf{Event-recording Automata.}
A $K$-Event-recording Automaton ($K$-ERA) $\A$ is a tuple $(Q,q_{\mathsf{init}},\Sigma,E,F)$ where $Q$ is a (non-empty) finite set of states, $q_{\sf init} \in Q$ is the initial state, $\Sigma$ is a finite alphabet, $E \subseteq Q \times \Sigma \times \SC{\Sigma, K} \times Q$ is the set of transitions, and $F \subseteq Q$ is the set of accepting states. Each transition in $\A$ is a tuple $(q, \sigma, g, q')$, where $q, q' \in Q$, $\sigma \in \Sigma$ and $g \in \SC{\Sigma, K}$, $g$ is called the \emph{guard} of the transition.
$\A$ is called a $K$-\emph{deterministic}-ERA ($K$-DERA) if for every state $q \in Q$ and every letter $\sigma$, if there exist two transitions $(q, \sigma, g_1, q_1)$ and $(q, \sigma, g_2, q_2)$ then $\sem{g_1} \cap \sem{g_2} = \emptyset$.
Again, we will not mention $K$ unless necessary.

\begin{figure}[t]
    \centering
    \begin{tikzpicture}
        \begin{scope}[every node/.style={circle, draw, inner sep=2pt,
            minimum size = 3mm, outer sep=3pt, thick}] 
            \node [double] (0) at (0,0) {$q_0$}; 
            \node [double] (1) at (3,0) {$q_1$};
            \node [double] (2) at (6,0) {$q_2$};
        \end{scope}
        \begin{scope}[->, thick]
            \draw (-0.7,0) to (0); 
            \draw (0) to (1);
            \draw (1) to (2);
            \draw [bend right = 60] (2) to (1);
        \end{scope}
        \node at (1.5,0.2) {\scriptsize $a$};
        \node at (4.5,0.2) {\scriptsize $b,~x_a = 1$};
        \node at (4.5,1.2) {\scriptsize $a,~x_b \leq 1$};
    \end{tikzpicture}
    \caption{A deterministic ERA describing the timed language, where every $a$ is followed by a $b$ after exactly $1$ time unit and every $b$ is followed by an $a$ within $1$ time unit. Here and in the successive figures, we do not write the guard $\top$.} 
    \label{fig:dera}
\end{figure}
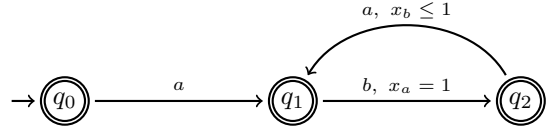

For the semantics, initially, all the clocks of an ERA start with the value $0$ and then they all increase with rate $1$. For every transition on a letter $\sigma$, once the transition is taken, its corresponding recording clock $x_\sigma$ gets reset to the value $0$. 

\medskip
\noindent\textbf{Clocked words.} \label{paragraph:clocked-words}
Given a timed word ${\sf tw}=(\sigma_1,t_1)(\sigma_2,t_2)\dots (\sigma_n,t_n)$ over $\Sigma$, we associate with it a \emph{clocked word} 
\label{timed-to-clocked-word}
$\mathsf{cw}(\mathsf{tw})=(\sigma_1,v_1)(\sigma_2,v_2)\dots (\sigma_n,v_n)$ where each $v_i : X_{\Sigma} \rightarrow \mathbb{R}_{\geq 0}$ maps each clock of $X_{\Sigma}$ to a real value as follows: $v_i(x_\sigma)=t_i-t_j$ where $j=\max\{ 1 \leq l < i \mid \sigma_l=\sigma \}$, with the convention that $\max(\emptyset)=0$ and $t_0 = 0$. 
In words, the value $v_i(x_\sigma)$ is the amount of time elapsed since the last occurrence of $\sigma$ in {\sf tw}; 
which is why we call the clocks $x_\sigma$ `recording' clocks. So, as mentioned earlier, even though not explicitly, each clock is implicitly reset immediately after every occurrence of its associated event.

\medskip
\noindent\textbf{The timed language of a $K$-ERA.}
A timed word $\mathsf{tw}=(\sigma_1,t_1)(\sigma_2,t_2)\dots (\sigma_n,t_n)$, with its clocked word 
${\sf cw}({\sf tw})=(\sigma_1,v_1)(\sigma_2,v_2)\dots (\sigma_n,v_n)$, is \emph{accepted} by $\A$ if there exists a sequence of states $q_0 q_1 \dots q_n$ of $\A$ such that $q_0=q_{\sf init}$, $q_n \in F$, and for all $1 \leq i \leq n$, there exists $e=(q_{i-1},\sigma,\psi,q_i) \in E$ such that $\sigma_i=\sigma$, and $v_i \models \psi$. 
The set of all timed words accepted by $\A$ is called the \emph{timed language} of $\A$, and will be denoted by $\ltw{\A}$. 

For example, the automaton $\Aa$ depicted in Figure~\ref{fig:dera}, accepts all the timed words where 
every $a$ is followed by a $b$ after exactly $1$ time-unit and every $b$ is followed by an $a$ within $1$ time unit.
In particular, $\A$ accepts $(a,2.3)(b,3.3)(a,3.4)$, but rejects $(a,2.3)(b,3.4)(a,3.4)$.

A timed language $L$ is $K$-ERA-definable if there exists a $K$-ERA $\Aa$ such that $\ltw{\Aa} = L$. Lemma~\ref{thm:era-equals-dera} implies that a timed language $L$ is $K$-ERA-definable iff it is also $K$-DERA-definable.

\begin{lemma}[\cite{AFH99}]
    \label{thm:era-equals-dera}
    For every $K$-ERA $\Aa$, there exists a $K$-DERA $\Aa'$ such that the two automata have the same timed language, $\ie$ $\ltw{\Aa} = \ltw{\Aa'}$. 
\end{lemma}

\smallskip
\noindent\textbf{Symbolic timed words} over ($\Sigma, K$) are finite sequences $(\sigma_1,g_1)(\sigma_2,g_2)\dots (\sigma_n,g_n)$ where each $\sigma_i \in \Sigma$ and $g_i \in \SC{\Sigma, K}$, for all $1 \leq i \leq n$.
Similarly, a \emph{region word} over ($\Sigma$, K) is a finite sequence $(\sigma_1,r_1)(\sigma_2,r_2)\dots (\sigma_n,r_n)$ where each $\sigma_i \in \Sigma$ and $r_i \in \SimC{\Sigma, K}$ \footnote{since simple constraints are similar (but not equal) to the standard notion of `regions' as defined in~\cite{AD94}, we refer to this kind of symbolic timed words as region words}, for all $1 \leq i \leq n$. 
We will use $\SW{\Sigma, K}$ and $\RW{\Sigma, K}$ to denote the sets of all symbolic timed words and all region words over an alphabet ($\Sigma, K$), respectively. 

We are now equipped to define when a timed word ${\sf tw}$ is compatible with a symbolic timed word $\sw$. Let  $\tw=(\sigma_1,t_1)(\sigma_2,t_2)\dots (\sigma_n,t_n)$ be a timed word and $\sw=(\sigma'_1,g_1)(\sigma'_2,g_2)\dots (\sigma'_m,g_m)$ be a symbolic timed word. We say that $\tw$ is compatible with $\sw$, noted $\tw \models \sw$ if we have that: $(i)$ $n=m$, \emph{i.e.}, both words have the same length, $(ii)$ $\sigma_i=\sigma'_i$ and $(iii)$ $v_i \models g_i$, for all $i$, $1 \leq i \leq n$. We use $\sem{\sw}$ to denote the set of all timed words that are compatible with $\sw$, that is, $\sem{\sw}=\{ \tw \in {\sf TW}(\Sigma) \mid \tw \models \sw\}$. We say that a symbolic timed word $\sw$ is \emph{consistent} if $\sem{\sw} \neq \emptyset$, and \emph{inconsistent} otherwise.


\section{A passive learning algorithm for ERA}
\label{sec:algo}

In this section, we propose a \emph{state-merging algorithm} -- we call it \leap\ (\textsf{L}earning \textsf{E}vent-recording \textsf{A}utomata \textsf{P}assively) -- that can be seen as an adaptation of the classical state-merging algorithm RPNI~\cite{oncina1992} for regular inference, into the timed setting.
Typically, in a  passive learning framework, a pair of sets of words $\S = (\splus, \sminus)$, called the \emph{sample set}, is given as input, and the objective is to construct an automaton $\A$ (an ERA in our case) that is \emph{consistent} with $\S$. Here \emph{consistency} means, $\A$ accepts all the words in $\splus$ and rejects all the words in $\sminus$. It is important to note here that no two words, one from $\splus$ and one from $\sminus$, should have a non-empty intersection.
Here we assume that the alphabet $\Sigma$ of the target ERA is
given as an input, and that, the sample sets contain symbolic timed words over the alphabet $\Sigma$.

\subsection{Order between symbolic words} \label{order-symbolic-words}
The set of all symbolic timed words can be ordered using a \emph{total order} `$\order$'. This order will be used later in this section for merging states in \leap, and also in Section~\ref{sec:leap-completeness-sdera} to prove the completeness result. 
We now define this order relation.

Let $(\Sigma,<)$ be a total order over the alphabet $\Sigma$. This order induces the following order on the clocks $X_{\Sigma}$: $x_{\sigma_1} < x_{\sigma_2}$ if and only if $\sigma_1 < \sigma_2$, for all $\sigma_1,\sigma_2 \in {\Sigma}$. Let $v_1$ and $v_2$ be two valuations for the clocks in $X_{\Sigma}$. We can order those two valuations using the following lexicographic order: $v_1 < v_2$ if and only if $\exists x_{\sigma_i} \in X_{\Sigma}$ s.t. $v_1(x_{\sigma_i}) < v_2(x_{\sigma_i})$, and for all $x_{\sigma_j} < x_{\sigma_i}$, $v_1(x_{\sigma_j})=v_2(x_{\sigma_j})$. 
Now, let us note that for two simple constraints $r_1$ and $r_2$, we have that either $\sem{r_1}=\sem{r_2}$ or $\sem{r_1} \cap \sem{r_2}=\emptyset$. From this observation, we can define a total order between elementary constraints (over say $x$), $r_1 \order r_2$ if for every $v_1 \in \sem{r_1}$ and every $v_2 \in \sem{r_2}$, $v_1(x) < v_2(x)$. We can extend this order to an order between simple constraints: $r_1 < r_2$ if there exists a clock $x_{\sigma_i}$ such that $r_1 {\downarrow_{i}} \order r_2{\downarrow_i}$ --where $r{\downarrow_i}$ denotes the projection of the simple constraint $r$ on the clock $x_{\sigma_i}$-- and for every $x_j < x_i$, $r_1{\downarrow_j} = r_2{\downarrow_j}$.
We can now extend this order to an order between two constraints
$g_1$ and $g_2$, in the following manner: $g_1 \order g_2$ if either (i) $\sem{g_1} \subset \sem{g_2}$, or (ii) there exists two simple constraints $r_1 \order r_2$ where $\sem{r_1} \subseteq (\sem{g_1} \setminus \sem{g_2})$ and $\sem{r_2} \subseteq (\sem{g_2} \setminus \sem{g_1})$. One can verify that this relation defined between  constraints is, indeed, a total order.

We define an ordering between two \emph{symbolic letters} as follows: 
$(\sigma_1, g_1) \order (\sigma_2, g_2)$ if either (i) $\sigma_1 < \sigma_2$ or (ii) $\sigma_1 = \sigma_2$ and $g_1 \order g_2$; where $g_1, g_2$ are two constraints.
Finally, we can extend this order to an order between two symbolic words. Let $\ssw_1 = (\sigma_1^1, g_1^1) \cdots (\sigma_k^1, g_k^1)$ and $\ssw_2 = (\sigma_1^2, g_1^2) \cdots (\sigma_l^2, g_l^2)$ be two symbolic words. Now, $\ssw_1 \order \ssw_2$ if:
\begin{enumerate}
    \item $k < l$, that is, $\ssw_1$ is shorter than $\ssw_2$, or 
    \item $k = l$ and $\exists 1 \le i \le n$ such that $(\sigma_i^1, g_i^1) \order (\sigma_i^2, g_i^2)$ and $\forall 1 \le j < i$, $(\sigma_j^1, g_j^1) = (\sigma_j^2, g_j^2)$.
\end{enumerate}

\begin{algorithm}[t]
\caption{Pseudo-code for algorithm \leap}
 \label{alg:LEAP}
	\begin{algorithmic}
		\Require An alphabet $\Sigma$, and a \emph{consistent} sample set ($\splus, \sminus$) 
		\Ensure An ERA that accepts (resp. rejects) all timed words in $\splus$ (resp. $\sminus$)
	\end{algorithmic}
	\begin{algorithmic}[1]
	    \State ${\tree}$ $\leftarrow$ $\PT$($\splus$)
	    \State $\textsc {Red} = \{q_{\varepsilon}\}$
            \State $\textsc {Blue} = \{q_u \mid q_u$ is a successor of $q_\varepsilon$\} \hfill // always sorted wrt a total order
     \While {\textsc{Blue} is not empty} 
         \State pop $q_v$ where $v$ is the minimum element wrt $\order$ in $\{v' \mid q_{v'} \in \textsc{Blue}\}$ 
         \For {every node $q_u \in~$\textsc{Red}}
            \If {merge($q_u$, $q_v$) is possible}
                \State add the successors of $q_u$ ($\notin$ \textsc{Blue}) to \textsc{Blue} \hfill// keeping \textsc{Blue} sorted
                \State go to the next iteration of the \textbf{while} loop
            \EndIf
	    \EndFor
        \If {no merge was possible}
            \State add $q_v$ to \textsc{Red}
            \State add the successors of $q_v$ to \textsc{Blue} \hfill// keeping \textsc{Blue} sorted
        \EndIf
     \EndWhile
	\end{algorithmic}
\end{algorithm}

\medskip
\subsection{An overview of the algorithm \leap} 
\label{subsec:leap-algo}
\leap\ initially constructs a tree, called the \emph{prefix tree}, denoted $\PT$, from the positive set $\splus$ and then tries to merge nodes from the tree keeping the resulting automaton \emph{consistent} with the sample. 
During the procedure, we maintain a coloring of nodes in the tree -- red and blue -- where the red nodes represent the nodes that will be the states in the target automaton, and at each iteration of the algorithm, a blue node will either be merged with 
one of the red nodes, or be promoted to a red node. A pseudo-code of \leap\ is presented in Algorithm~\ref{alg:LEAP}.
The algorithm consists of two main steps:

\medskip
\noindent\textbf{In the first step} 
\label{def:PTA} \leap~constructs a prefix tree $\tree$ from $\splus$ (line $1$). We can formally define the prefix tree as an ERA $\PT(\splus) = (Q,  q_\varepsilon, \Sigma, E, F)$ where the set of states $Q = \{q_u\mid u \in \pref{\splus}\}$ is the set of prefixes of words in $\splus$ (denoted by $\pref{\splus}$), the initial state is the state corresponding to the empty string $\varepsilon$, the set of edges $E$ is defined as follows: for $(a,g) \in \Sigma \times \SC{\Sigma}$, and $u, u.(a,g)\in \pref{\splus}$, the transition $(q_u, a, g, q_{u.(a,g)}) \in E$; and finally, the set of accepting states $F = \{q_u \mid u \in \splus\}$ is the set of all words in $\splus$, \emph{i.e.}, the prefix tree accepts exactly the words in $\splus$.

Note that there may exist states from which there are more than one transition on the same letter $\sigma$ but with different guards $g, g'$ that intersect. As a consequence, a prefix tree can be non-deterministic. However, note that, when all words in the sample set are region words, the $\PT$ will be deterministic (cf. Lemma~\ref{lem:rws-dont-intersect}).
    
\medskip
\noindent\textbf{In the second step}
\label{def:merge-fold} 
the algorithm initializes a set {\sc Red} that contains the initial node $q_\varepsilon$ (line~2), and  maintains a \emph{sorted list} {\sc Blue} containing the successors of $q_\varepsilon$ in the $\PT$ (line~3). 
The list {\sc Blue} is sorted according to the order $\order$ defined between symbolic timed words mentioned earlier in this section, the lowest element in the order being at the beginning of the list.
Then, it iteratively does the following until {\sc Blue} becomes empty: at every iteration $i$, the algorithm pops the first node $q_v$ from {\sc Blue} (line 5), and executes the procedure {\sf merge} that checks if a merge is possible between $q_v$ and one of the nodes in {\sc Red} (line 6-11): 
a merge is possible if the resulting automaton after the merge remains \emph{consistent} with the sample -- that is, if the merged automaton does not intersect with any word in $\sminus$; otherwise, the merge is discarded. 
When a merge is successful, {\sc Red} and {\sc Blue} are updated accordingly (line 8).
If no merge is possible, $q_v$ is \emph{promoted} to a red state, \ie, $\textsc{Red} := \textsc{Red} \uplus \{q_v\}$, and all successors of $q_v$ are appended to {\sc Blue} according to the order $\order$ (line 12-15). 
    
\subparagraph*{Procedure merge:}
    The procedure takes as input an ERA $\A$, a node $q_u$ from {\sc Red} and a node $q_v$ from {\sc Blue}. 
    Note that, every node in {\sc Blue} has a \emph{unique} predecessor node, and this node is always in {\sc Red}.
    Let $q_w$ be the unique predecessor of the blue node $q_v$, \ie, there exists a transition $t = q_w \xrightarrow{(\sigma,g)}q_v$ in $\A$. Then construct an ERA $\A'$ as follows: delete the transition $t$ from $\A$, and add the transition $q_w \xrightarrow{(\sigma, g)} q_u$ (if it does not already exist in $\A$), and then recursively
    \emph{fold} the subtree of $q_v$ into $q_u$. Here, \emph{fold} means after merging $q_u$ and $q_v$, if these nodes have successors $q_u'$ and $q_v'$ on transitions labelled with the \emph{same} letter $\sigma$ and the \emph{same} guard $g$, respectively, then we also merge $q_u'$ and $q_v'$, and repeat this process.
    Notice that, the automaton $\A'$ can, in general, be non-deterministic.

\subparagraph*{Checking consistency with the sample:}
\label{smt-encoding}

A merge between two states is allowed in \leap\ only if the resulting automaton does not intersect with any negative sample. We now provide an SMT-based algorithm for checking if an ERA intersects a symbolic word. Thus, this algorithm can be used inside \leap\ to determine when a merge is allowed: each time two states are merged during an iteration of \leap, say $\Aa$ is the resulting automaton due to the merge, then the following algorithm can be used to check if for every $w \in \sminus$, $\lang(\Aa) \cap \sem{w} \neq \emptyset$ or not. 
Such an encoding of this problem into an SMT formula is justified, since in the next section (in Lemma~\ref{prop:lang-incl-np-hard}) we will show that deciding this problem is, in fact, $\np$-complete.

We reduce the afore-mentioned problem in the theory of (quantifier-free) Linear Real Arithmetic (LRA) with Boolean variables, for which checking satisfiability is also known to be $\np$-complete. 
Assume $w$ be the symbolic word $(\sigma_1, g_1) (\sigma_2, g_2) \ldots (\sigma_p, g_p)$. We want to find a timed word (as a certificate) $w_t = (\sigma_1, t_1) (\sigma_2, t_2) \ldots (\sigma_p, t_p)$, where $t_i \in \R_{\ge 0}$ and $t_i \le t_j$ for every $1 \le i < j \le p$, such that $w_t \in \sem{w} \cap \lang(\Aa)$. For this, we construct a formula $\phi$ over $p$ real variables $t_1, t_2, \ldots, t_p$ such that, a model $w_t \models \phi$ iff $w_t \in \sem{w} \cap \lang(\Aa)$. The formula $\phi$ is constructed as a conjunction of $\phi_{w}$ and $\phi_{\Aa}$, as described below.

First, we encode the constraints present in $w$. Consider the $i$-th position in $w$, suppose the constraint in this position is $g_i$. Now, let $x_\sigma \sim c$ be an elementary constraint present as a conjunct in $g_i$.  Now, let $0 \leq j < i$ be the largest position in $w$ such that $\sigma_j = \sigma$. Then, $g_i$ will be satisfied by the timed word $w_t$ only if $t_i - t_j \sim c$. For every constraint $g_i$ in $w$, we construct the formula $\phi_{g_i}$ where each clock $x_\sigma$ is replaced with $t_i - t_j$. Then define $\phi_{w} := \bigwedge_{1 \leq i \leq n} \phi_{g_i}$.

We now describe the formula $\phi_{\Aa}$. We encode the underlying transition system $\Aa$ in a fairly standard way. 
Assuming $\Aa$ contains $n$ states, we use $\log(n)$-many Boolean variables to denote the states of the automaton. For each state $q$ of $\Aa$ we have the following formula:
    $\phi_{q} := \ell_1 \wedge \ell_2 \wedge \ldots \wedge \ell_k$, where $k = \log(n)$ and $\ell_j \in \{p_j, \neg p_j\}, 1 \le j \le k$.
The formula for the initial state is $\phi_I := \phi_{q_{\mathsf{init}}}$, and the formula for the final states is:
    $\phi_F := \bigvee_{q_f \in F} \phi_{q_f}$.
We now encode the transitions of $\Aa$.
Suppose $e = (q, \sigma, g, q')$ be a transition of $\Aa$, where $\sigma \in \Sigma$ and the guard $g$ is a constraint $\bigwedge_m x_m \sim c$. To encode the guards, we use the same technique as we used when encoding a symbolic word. 
We then get the formula for a transition:
$\phi_e := \phi_{q} \wedge \phi_{g} \wedge \phi_{q'}$. 
For the set of all transitions $E$ with event $\sigma$ of $\Aa$, noted $E|\sigma$, we use the formula $\phi_{E|\sigma} := \bigvee_{e = (q,\sigma,g,q') \in E} \phi_e$.

We will use superscripts on the state-variables ($\ell_i$'s) to denote the state of the automaton after $m$ steps in a run; for instance, $\phi_{q_i}^m := \ell_1^m \wedge \ell_2^m \wedge \ldots \wedge \ell_k^m$ will encode the fact that the automaton is in $q_i$ at the $m$-th step of a run. Further, the $m$-th transition in a run will be represented by $\phi_e^{m} := \phi_{q_i}^m \wedge \phi_{g_i} \wedge \phi_{q_j}^{m+1}$.
Similarly, we use $\phi_{E|\sigma}^m := \bigvee_{e = (q,\sigma,g,q') \in E} \phi_e^m$.
To represent the set of all $p$-length paths in $\Aa$ on $w$, we use the following formula:
\begin{equation*}
    \phi_{\Aa} := \phi_I^1 \wedge \phi_{E|\sigma_1}^1 \wedge \phi_{E|\sigma_2}^2 \ldots \wedge \phi_{E|\sigma_p}^p \wedge \phi_F^p
\end{equation*}
Note that, every satisfying model of $\phi_{\Aa}$ is a timed word $w_t$ that has a run in $\Aa$ along $w$, and further it starts from the initial state and ends at a final state after $p$ steps; i.e. $w_t$ has an \emph{accepting} run in $\Aa$.

Then the final formula is $\phi := \phi_{w} \wedge \phi_{\Aa}$. 

\begin{lemma}
    $\phi$ is satisfiable iff $\sem{w} \cap \lang(\Aa) \neq \emptyset$.
\end{lemma}

This concludes the description of the SMT encoding. 
Then one can show the correctness of \leap:

\begin{restatable}[Correctness]{theorem}{leapcorr}
\label{th:leap-correctness}
    Given a sample set $\mathsf{S} = (\splus, \sminus)$ of symbolic timed words as input, (1) \leap\ terminates, and (2) \leap\ returns a possibly nondeterministic automaton $\Aa$ such that for every $w \in \splus$, $\sem{w} \subseteq \lang(\Aa)$ and for every $w \in \sminus$, $\sem{w} \cap \lang(\Aa) \neq \emptyset$.
\end{restatable}

\begin{proof}[Proof (sketch)]
    (1) The prefix tree has finitely many nodes -- at most one for each prefix in $\pref{\splus}$. Then at every step of \leap, it either successfully merges two nodes, in which case, the number of nodes in the automaton reduces at least by one, or it promotes a node to red node. Since red nodes are never processed again in the future, this procedure terminates.

    (2) The algorithm is correct by construction, since after every merge we do not delete any node or transition present in the automaton, and moreover, when a merge is successful, it means that the resulting automaton is also consistent with the negative samples $\sminus$ as well.
\end{proof}

We will now illustrate the execution of \leap\ on an example.

\begin{example}
    \label{appendix:examplerun}

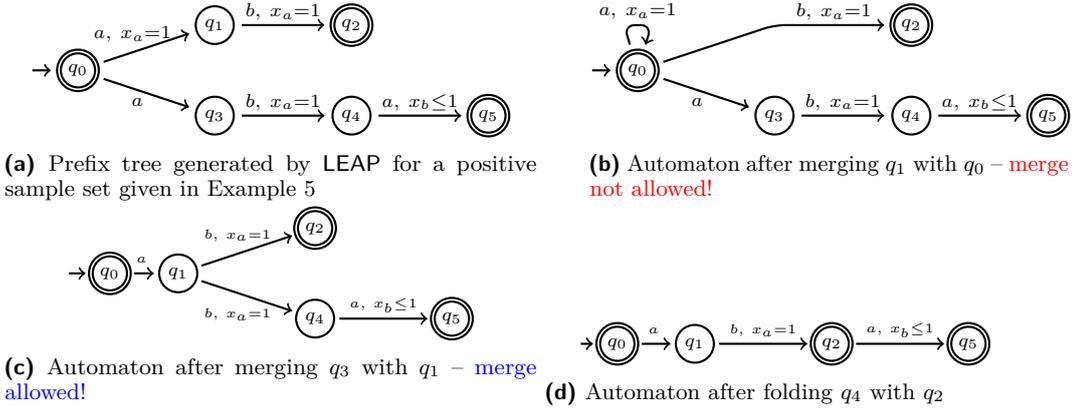
\begin{figure}[t]
\begin{subfigure}{0.5\textwidth}
    \centering
    \begin{tikzpicture}[scale=0.6]
    \everymath{\scriptstyle}
        \begin{scope}[every node/.style={circle, draw, inner sep=2pt,
            minimum size = 3mm, outer sep=3pt, thick}] 
            \node [double] (0) at (0,0) {$q_0$}; 
            \node [] (1) at (3,1) {$q_1$};
            \node [] (3) at (3,-1) {$q_3$};
            \node [double] (2) at (6,1) {$q_2$};
            \node [] (4) at (6,-1) {$q_4$};
            \node [double] (5) at (9,-1) {$q_5$};
        \end{scope}
        \begin{scope}[->, thick]
            \draw (-1,0) to (0); 
            \draw (0) to (1);
            \draw (1) to (2);
            \draw (0) to (3);
            \draw (3) to (4);
            \draw (4) to (5);
        \end{scope}
        \node at (1.2,0.8) { $a,~x_a = 1$};
        \node at (1.3,-0.7) { $a$};
        \node at (4.5,1.3) { $b,~x_a = 1$};
        \node at (4.5,-0.7) { $b,~x_a = 1$};
        \node at (7.5,-0.7) { $a,~x_b \leq 1$};
    \end{tikzpicture}
\caption{Prefix tree generated by \leap\ for a positive sample set given in Example~\ref{appendix:examplerun}}
\label{fig:running-ex-pta}
\end{subfigure}
\begin{subfigure}{0.45\textwidth}
    \centering
    \begin{tikzpicture}[scale=0.6]
        \everymath{\scriptstyle}
        \begin{scope}[every node/.style={circle, draw, inner sep=2pt,
            minimum size = 3mm, outer sep=3pt, thick}] 
            \node [double] (0) at (0,0) {$q_0$}; 
            \node [] (3) at (3,-1) {$q_3$};
            \node [double] (2) at (6,1) {$q_2$};
            \node [] (4) at (6,-1) {$q_4$};
            \node [double] (5) at (9,-1) {$q_5$};
        \end{scope}
        \begin{scope}[->, thick]
            \draw (-1,0) to (0); 
            \draw [rounded corners] (0) to (-0.3,1) to (0.3,1) to (0);
            \draw [rounded corners] (0) to (3,1) to (2);
            \draw (0) to (3);
            \draw (3) to (4);
            \draw (4) to (5);
        \end{scope}
        \node at (0,1.3) { $a,~x_a = 1$};
        \node at (1.3,-0.7) { $a$};
        \node at (4.3,1.3) { $b,~x_a = 1$};
        \node at (4.5,-0.7) { $b,~x_a = 1$};
        \node at (7.5,-0.7) { $a,~x_b \leq 1$};
    \end{tikzpicture}
    \caption{Automaton after merging $q_1$ with $q_0$ -- \textcolor{red}{merge not allowed!}}
    \label{fig:running-ex-merge1}
\end{subfigure}
\begin{subfigure}{0.5\textwidth}
    \centering
    \begin{tikzpicture}[scale=0.6]
        \everymath{\scriptstyle}
        \begin{scope}[every node/.style={circle, draw, inner sep=2pt,
            minimum size = 3mm, outer sep=2pt, thick}] 
            \node [double] (0) at (1.5,0) {$q_0$}; 
            \node [] (1) at (3,0) {$q_1$};
            \node [double] (2) at (6,1) {$q_2$};
            \node [] (4) at (6,-1) {$q_4$};
            \node [double] (5) at (9,-1) {$q_5$};
        \end{scope}
        \begin{scope}[->, thick]
            \draw (0.6,0) to (0); 
            \draw [rounded corners] (0) to (1);
            \draw (1) to (2);
            \draw (1) to (4);
            \draw (4) to (5);
        \end{scope}
        \node at (2.2,0.3) {\scriptsize $a$};
        \node at (4.3,0.85) {\scriptsize $b,~x_a = 1$};
        \node at (4.3,-0.9) {\scriptsize $b,~x_a = 1$};
        \node at (7.5,-0.7) {\scriptsize $a,~x_b \leq 1$};
    \end{tikzpicture}
    \caption{Automaton after merging $q_3$ with $q_1$ -- \textcolor{blue}{merge allowed!}}
    \label{fig:running-ex-merge2}
    \end{subfigure}
    \begin{subfigure}{0.45\textwidth}
    \centering
    \begin{tikzpicture}[scale=0.6]
    \everymath{\scriptstyle}
        \begin{scope}[every node/.style={circle, draw, inner sep=2pt,
            minimum size = 3mm, outer sep=2pt, thick}] 
            \node [double] (0) at (1.3,0) {$q_0$}; 
            \node [] (1) at (3,0) {$q_1$};
            \node [double] (2) at (6,0) {$q_2$};
            \node [double] (5) at (9,0) {$q_5$};
        \end{scope}
        \begin{scope}[->, thick]
            \draw (0.5,0) to (0); 
            \draw [rounded corners] (0) to (1);
            \draw (1) to (2);
            \draw (2) to (5);
        \end{scope}
        \node at (2.1,0.3) {\scriptsize $a$};
        \node at (4.5,0.3) {\scriptsize $b,~x_a = 1$};
        \node at (7.5,0.3) {\scriptsize $a,~x_b \leq 1$};
    \end{tikzpicture}
    \caption{Automaton after  folding $q_4$ with $q_2$}
    \label{fig:running-ex-merge3}
    \end{subfigure}
    \caption{Different steps of the algorithm \leap.
    In the prefix tree, we associate each state with the word that leads to that state as described in Section~\ref{def:PTA}, for example, $q_2$ corresponds to $(a,~x_a = 1)(b,~x_a = 1)$, and $q_3$ corresponds to $(a,\top)$, etc.
    }
    \label{fig:running-ex-illustration}
\end{figure}

Let $S = (\splus, \sminus)$ be a sample set, where -- 
\begin{align*}
    \splus := &\{\varepsilon, (a,~x_a=1)(b,~x_a = 1), (a,\top)(b,~x_a = 1)(a,~x_b \leq 1)\} \\
    \sminus := &\{(a,\top)(a,\top), (a,\top)(b,~x_a = 1)(a,~x_b = 2)(b,~x_a = 1), \\
    & (a,\top)(b,~x_a = 1)(b,~x_a = 1), (a,\top)(b,~x_a = 1)(a,~x_b = 1)(a,\top)(b,~x_a = 1)\}
\end{align*}
With this sample as input, in the first step, the algorithm computes the prefix tree depicted in Figure~\ref{fig:running-ex-pta}.
The algorithm starts by putting $q_0$ in \textsc{Red} and $q_1$ and $q_3$ in \textsc{Blue} according to the order $\order$, mentioned earlier in this section. It picks $q_1$ and tries to merge this state with $q_0$, the merged automaton is depicted in Figure~\ref{fig:running-ex-merge1}, but the merge fails since it accepts the negative word $(a,\top)(a,\top)$.
Since $q_0$ is the only state in \textsc{Red}, $q_1$ gets promoted to \textsc{Red}, and $q_2$ is added to \textsc{Blue}, keeping the list sorted. 
The algorithm then picks $q_3$; the merge of $q_3$ with $q_0$ fails, but, the merge of $q_3$ with $q_1$ succeeds. The automaton computed after this merge is depicted in Figure~\ref{fig:running-ex-merge2}. After this, the algorithm \emph{folds} $q_2$ and $q_4$ (since the incoming transitions to these two states are both from $q_1$ with the same event and guard) and gets the automaton in Figure~\ref{fig:running-ex-merge3}. 
At this point, the only state in \textsc{Blue} is $q_2$. 
Both merges of $q_2$ with $q_0$ and $q_2$ with $q_1$ fail, and therefore,
$q_2$ is promoted to \textsc{Red}, and $q_5$ to \textsc{Blue}. Finally, the algorithm tries to merge $q_5$ with $q_0$, which fails due to the negative word $(a,\top)(b,~x_a = 1)(a,~x_b = 1)(a,\top)(b,~x_a = 1)$. The merge of $q_5$ with $q_1$ succeeds, the algorithm computes the automaton in Figure~\ref{fig:dera}. 
\end{example}

The above example demonstrates that our algorithm \leap\ does not need to maintain deterministic structure during merging, and the result of the merging phase can be a non-deterministic ERA. This is in sharp contrast to RPNI. 
We now provide a more detailed comparison of \leap\ with RPNI and then compare with other state-merging algorithms that we are aware of.

\subsection{Comparing LEAP with RPNI}
\label{sec:leap-vs-rpni}
While the overall idea behind the algorithm \leap\ is broadly similar to RPNI, we mention below a few key differences between them.

\medskip
    \noindent \textbf{Role of semantics.} 
    It might feel tempting to think that, one can 
    simply interpret the symbolic timed words over $\Sigma$ in the sample set as words from a modified alphabet $\Sigma'=\Sigma \times C_S$, where $C_S$ denotes the set of all constraints present in the sample set, and apply techniques from classical regular inference. 
    However, in the timed setting, this does not align with the properties of a sample set. For instance, consider the following sample set $S = (\splus, \sminus)$ over alphabet $\Sigma = \{a,b\}$ with $\splus = (a,\top)$ and $\sminus = (a,x_b\le 1)$. Then this sample set is indeed consistent w.r.t. the alphabet $\Sigma' = \{(a,\top), (a,x_b\le1)\}$ in the regular case, however, this is not consistent if we consider timing aspects of the words. This is why one must be able to deal with the semantics of such sets, rather than only syntax.
    
\medskip
    \noindent\textbf{Comparing merge.} First notice that, at every iteration of the algorithm, one needs to check if a \emph{merge} between two states of an automaton is successful, that is,
    if the resulting automaton does not intersect with any of the symbolic timed words from the negative sample ($\sminus$). 
    More formally, let $A$ be the ERA after a merge along the procedure, then one has to check if for all $w \in \sminus$, $\sem{w} \cap L(A) = \emptyset$. 
    In RPNI, this check can be done in polynomial time as it amounts to check membership of a word in a DFA. However, for symbolic timed words and ERA, this check is computationally more expensive as we show in Section~\ref{sec:complexity}: given a DERA $A$ and a symbolic timed word $w$, checking if $\sem{w} \cap L(A) \neq \emptyset$ is $\np$-complete.
    
    Secondly, unlike RPNI, due to the timing constraints in an ERA, the automaton produced at each step of \leap, and hence the output ERA, can in general be non-deterministic. This helps us learn potentially smaller automata, compared to the existing learning algorithms for this class, as will be illustrated in Section~\ref{sec:experiments}.
   
\medskip
    \noindent\textbf{Towards completeness.} The idea of language identification from given data, introduced by Gold~\cite{DBLP:journals/iandc/Gold78}, is a desired property of a passive learning algorithm. To this end, given any language $L$ from a particular class of languages (for example, the class of regular languages), 
    one defines \emph{characteristic sets} $S_L$, so that if any sample set $S$ \emph{containing} $S_L$ is given as input, the learning algorithm is guaranteed to return an automaton accepting $L$. It is well-known that such characteristic sets exist for RPNI for the class of regular languages.    

    The construction of a characteristic set, in the case of regular languages, depends on the Myhill-Nerode equivalence classes of the language. Intuitively, this can also be obtained
    from the minimal DFA recognizing the language: for every transition $t$ present in the minimal DFA, such a set must contain a word $w \in \splus$ that visits the transition $t$, and further, for every state $q$ and every transition $t$ that leads to some $q'(\ne q)$, the set must contain two words that \emph{distinguishes} $q$ with $t$: $\exists w_1.z \in \splus$ and $w_2.z \in \sminus$ s.t. either $w_1$ leads to $q$ and $w_2$ leads to $q'$ (via $t$), or vice-versa. In Section~\ref{sec:leap-completeness-sdera}, we will show a similar completeness result for the class of ERA-recognizable languages w.r.t. \leap.
    However, here we take a different approach, that only depends on the language to be learned -- and not the automaton -- to prove completeness, which however can easily be adapted for the regular case.

\subsection{Comparing \leap\ with other state-merging algorithms}
\label{subsec:comparison-state-merging}
As we have remarked in the introduction, state-merging techniques have been used for inferring ERA and an extended class of ERA in an active-learning framework~\cite{GJP06,HJM20}. The state-merging criterion proposed in~\cite{HJM20} is an adaptation of the criterion proposed in~\cite{GJP06} to handle the larger class. Since in this article we are only dealing with ERA, we compare here the merging criterion proposed in Section~\ref{subsec:leap-algo} with the one proposed in~\cite{GJP06}. 
Intuitively, their algorithm is targeting a fixed language, hence their merging criterion is more restricted and thus less suited for a passive learning algorithm. 
More precisely, their merging criterion relies on the fact that the timed-decision tree contains `complete' information about the words (upto a fixed length) of the target language.
We, on the other hand, only have access to a fixed set of words. Consequently, their merging criteria are more restricted, and thus less suited for a passive learning algorithm. Below, we show through an example, that if we follow the merging criteria of~\cite{GJP06}, then we will, in general, merge less number of nodes.  

\begin{figure}[t]
    \centering
    \begin{tikzpicture}[scale=0.6]
    \everymath{\scriptstyle}
        \begin{scope}[every node/.style={circle, draw, inner sep=2pt,
            minimum size = 3mm, outer sep=3pt, thick}] 
            \node [] (0) at (0,0) {$q_0$}; 
            \node [] (1) at (3,1) {$q_1$};
            \node [] (3) at (3,-1) {$q_3$};
            \node [double] (2) at (6,1) {$q_2$};
            \node [double] (4) at (6,-1) {$q_4$};
        \end{scope}
        \begin{scope}[->, thick]
            \draw (-1,0) to (0); 
            \draw (0) to (1);
            \draw (1) to (2);
            \draw (0) to (3);
            \draw (3) to (4);
        \end{scope}
        \node at (1.2,0.8) { $a,~x_a > 1$};
        \node at (1.3,-0.9) { $b, ~x_b > 1$};
        \node at (4.5,1.3) { $c,~x_a < 3$};
        \node at (4.5,-0.7) { $d,~x_b < 3$};
    \end{tikzpicture}
\caption{Prefix tree generated by \leap}
\label{fig:leap-vs-gjp}
\end{figure}
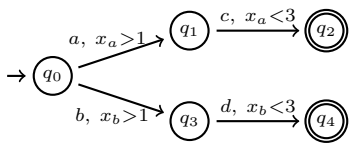

Consider the sample $\splus = \{(a,x_a > 1)(c, x_a<3), (b, x_b>1)(d, x_b<3)\}$, $\sminus = \{(a,\top)(a,\top), (b,\top)(b,\top)\}$ (the prefix tree constructed by \leap\ for this sample set is depicted in Figure~\ref{fig:leap-vs-gjp}). Here, the nodes $q_1$ and $q_3$ are not merged with the node $q_0$ due to the two negative words present in the sample, respectively; on the other hand, the nodes $q_1$ and $q_3$ are merged by \leap. However, $q_1$ and $q_3$ will not be merged if we follow the criteria of~\cite{GJP06}. This is because, (i) the set of paths from $q_1$ (resp. $q_3$) are not included in the set of paths from $q_3$ (resp. $q_1$), and (ii) the constraint on clocks induced by the paths leading to $q_1$ and $q_3$ are different -- which are the two criteria proposed in~\cite{GJP06} for merging two nodes.

\section{Complexity results}
\label{sec:complexity}

In this section, we will first show that given an ERA $\Aa$ and a symbolic timed word $w$, checking whether the language accepted by $\Aa$ \emph{is not disjoint from} $w$, \emph{i.e.} whether $\lang(\Aa) \cap \sem{w} \neq \emptyset$, is $\np$-complete. Consequently, deciding if a merge is allowed in any iteration of \leap\ is also $\np$-complete. Later, we will discuss the overall complexity of the \leap\ algorithm.

Given an ERA $\A$ and a symbolic timed word $w$, we study the complexity of the following decision problem:

  \begin{fmpage}{0.9\linewidth}
  \label{intersection-non-emptiness-problem}
  \textsc{Intersection non-emptiness problem} \\
  {\bf Input}: An ERA $\A$, and a symbolic timed word $w$.\\
  {\bf Output}: $\mathsf{Yes}$ if $\lang(\A) \cap \sem{w} \neq \emptyset$, and $\mathsf{No}$ otherwise.
  \end{fmpage}

We show that the above problem is $\np$-complete, even for DERA. We prove the hardness by providing a polynomial time reduction from $3$-\textsf{SAT}.

\begin{lemma}
\label{prop:lang-incl-np-hard}
The \textsc{Intersection non-emptiness} problem for DERA is $\np$-hard.
\end{lemma}

\begin{proof}
The proof of $\np$-hardness is established via a polynomial-time reduction from $3$-\textsf{SAT}, which is known to be $\np$-complete~\cite{Cook71}. Given a $3$-CNF instance $\varphi$, we construct a DERA $\Aa_{\varphi}$ and a symbolic timed word $w_{\varphi}$ such that $\varphi$ is satisfiable if and only if $\lang(\Aa_{\varphi}) \cap \sem{w_{\varphi}}$ is non-empty.

The high-level idea is as follows: the automaton $\Aa_{\varphi}$ consists of two parts. The first part includes a block of states for each propositional variable in $\varphi$, ensuring that every run of $\Aa_{\varphi}$ corresponds to an assignment of these variables. The second part ensures that if a run of $\Aa_{\varphi}$ completes the first part, the corresponding assignment must make at least one literal in each clause of $\varphi$ true. Together, they ensure that every run of $\Aa_{\varphi}$ corresponds to a satisfying assignment for $\varphi$.
Then we construct $w_{\varphi}$ in such a way that, for every timed word $w_t$ that has an accepting run in $\Aa_\varphi$, we will have $w_t \in \sem{w_\varphi}$. This in turn will imply that $\varphi$ is satisfiable if and only if $ \lang(\A_{\varphi}) \cap \sem{w_{\varphi}} \neq \emptyset$. 
We formalize this idea below. An example of the construction of $\Aa_\varphi$ is depicted in Figure~\ref{fig:np-hard-example}.

We first fix some notations: let $\varphi$ be the $3$-CNF formula $C_1 \wedge C_2 \wedge \ldots \wedge C_m$, where each $C_h$ is of the form $\ell_{h,1} \vee \ell_{h,2} \vee \ell_{h,3}$. We call each $C_h$ a \emph{clause} and each $\ell_{h,j}$ a \emph{literal}. We assume $p_1,p_2, \ldots, p_n$  are all the propositional variables present in~$\varphi$. W.l.o.g., the order on the variables is fixed throughout the construction.
A literal is then either a propositional variable or the negation of a propositional variable, \emph{i.e.}, for every $1 \leq j \leq 3$,
  $\ell_{h,j} \in \{p_i, \neg p_i \mid 1 \le i \le n\}$.
\medskip

\noindent\textbf{The automaton $\Aa_{\varphi}$.} 
The alphabet $\Sigma$ of $\Aa_{\varphi}$ consists of $p_i$ for every $1 \le i \le n$ and two special letters $\delim$ and $\ok$. Since $\Aa_{\varphi}$ is an ERA, we have a clock $x_\sigma$ for every $\sigma \in \Sigma$.
Intuitively, the amount of delay required before taking the transitions labelled with $p_i$'s will be used to determine the values of the propositions,
the transitions on $\delim$ will maintain the time elapsed during an execution, and the transitions on $\ok$ will be used to ensure that every clause evaluates to $\tt$.

Formally,
$\A_{\varphi} = (Q, \Sigma, q_0, \Delta, F)$, 
where:
$Q = \{q_0\} \cup \{v_i \mid 0\le i \le n\} \cup \{q_{i},\bar{q}_i \mid 1\le i \le n\} \cup \{C_{h} \mid 1\le h \le m+1\}$;
 $\Sigma = \{p_i \mid 1\le i \le n\} \cup \{\delim\} \cup \{\ok\}$; and
 $q_0 \in Q$ is the initial state. 
Recall here that $n$ is the number of propositional variables and $m$ is the number of clauses present in $\varphi$.
\begin{figure}[t]
        \centering
        \begin{tikzpicture}[scale = 0.8]
            \everymath{\scriptstyle}
            \begin{scope}[every node/.style={circle, draw, inner sep=2pt,
                minimum size = 3mm, outer sep=2pt, thick}] 
                \node (-1) at (-1.3,0) {$q_0$};

                \node (0) at (0,0) {$\scriptscriptstyle v_0$};
                
                \node (1) at (1,1) {$q_1$};
                \node (2) at (2,-1) {$\bar{q}_1$};
                
                \node (3) at (3,0) {$\scriptscriptstyle v_1$};
                
                \node (4) at (4,1) {$q_2$};
                \node (5) at (5,-1) {$\bar{q}_2$};
                
                \node (6) at (6,0) {$\scriptscriptstyle v_2$};
                
                \node (7) at (7,1) {$q_3$};
                \node (8) at (8,-1) {$\bar{q}_3$};
                
                \node (9) at (9,0) {$\scriptscriptstyle v_3$};
                \node (10) at (10,1) {$q_4$};
                \node (11) at (11,-1) {$\bar{q}_4$};
                \node (12) at (12,0) {$\scriptscriptstyle v_4$};

                \node (C1) at (0,-4) {$C_1$};
                \node (C2) at (4,-4) {$C_2$};
                \node (C3) at (8,-4) {$C_3$};
                \node [accepting, outer sep=2.5pt] (fin) at (12,-4) {$C_4$};
            \end{scope}
            
            \begin{scope}[->, thick]
             \draw (-1.3,1) to (-1); 
                \draw (-1) to  node [sloped,midway,above] {$\delim$}(0); 
                
                \draw [ color=red] (0) to  node[midway,above,sloped] {$p_1$}
                node [ midway, below, sloped] {$x_\delim = 1$}  (1);
                
                \draw [ color=blue] (0) to node[midway,above,sloped] {$p_1$}
                node [ midway, below, sloped] {$x_\delim = 2$}  (2);
                
                \draw [] (1) to node[midway,above,sloped] {$\delim$}
                node [ midway, below, sloped] {$x_{p_1} = 2$} (3);
                
                \draw [] (2) to node[midway,above,sloped] {$\delim$}
                node [ midway, below, sloped] {$x_{p_1} = 1$} (3);
                
                \draw [color=red] (3) to node[midway,above,sloped] {$p_2$}
                node [ midway, below, sloped] {$x_\delim = 1$} (4);
                
                \draw [ color=blue] (3) to node[midway,above,sloped] {$p_2$}
                node [ midway, below, sloped] {$x_\delim = 2$} (5);
                
                \draw [] (4) to node[midway,above,sloped] {$\delim$}
                node [ midway, below, sloped] {$x_{p_2} = 2$} (6);
                
                \draw [] (5) to node[midway,above,sloped] {$\delim$}
                node [ midway, below, sloped] {$x_{p_2} = 1$} (6);
                
                \draw [ color=red] (6)  to node[midway,above,sloped] {$p_3$}
                node [ midway, below, sloped]{$x_{\delim} = 1$} (7);
                
                \draw [ color=blue] (6)  to node[midway,above,sloped] {$p_3$}
                node [ midway, below, sloped]{$x_{\delim} = 2$}(8);
                
                \draw [] (7) to node[midway,above,sloped] {$\delim$}
                node [ midway, below, sloped] {$x_{p_3} = 2$}  (9);
                
                \draw [] (8) to node[midway,above,sloped] {$\delim$}
                node [ midway, below, sloped] {$x_{p_3} = 1$}  (9);
                
                \draw [color = red] (9) to node[midway,above,sloped] {$p_4$}
                node [ midway, below, sloped]{$x_{\delim} = 1$} (10);
                
                \draw [color = blue] (9) to node[midway,above,sloped] {$p_4$}
                node [ midway, below, sloped]{$x_{\delim} = 2$} (11);
                
                \draw [] (10) to node[midway,above,sloped] {$\delim$}
                node [ midway, below, sloped]{$x_{p_4} = 2$} (12);
                
                \draw [] (11) to node[midway,above,sloped] {$\delim$}
                node [ midway, below, sloped]{$x_{p_4} = 1$} (12);
                
                \draw [rounded corners, color=green!60!black] (C1) to (1,-3) to node [midway,above] {$\ok$} node[midway,below]{$x_{p_1} = 11 \wedge x_\delim = 0$} (3,-3) to (C2); 
                \draw [color=green!60!black]  (C1) to node [midway,above] {$\ok$} node [midway,below] {$x_{p_2} = 8 \wedge x_\delim = 0$}(C2); 
                \draw [rounded corners,color=green!60!black] (C1) to (1,-5) to node [midway,above] {$\ok$} node[midway,below]{$x_{p_3} = 4 \wedge x_\delim = 0$} (3,-5) to (C2);
                
                \draw [rounded corners, color=green!60!black] (C2) to (5,-3) to node [midway,above] {$\ok$} node[midway,below]{$x_{p_2} = 7 \wedge x_\delim = 0$} (7,-3) to (C3); 
                \draw [color=green!60!black]  (C2) to node [midway,above] {$\ok$} node [midway,below] {$x_{p_3} = 4 \wedge x_\delim = 0$}(C3); 
                \draw [rounded corners, color=green!60!black] (C2) to (5,-5) to node [midway,above] {$\ok$} node[midway,below]{$x_{p_4} = 2 \wedge x_\delim = 0$} (7,-5) to (C3);

                \draw [rounded corners, color=green!60!black] (C3) to (9,-3) to node [midway,above] {$\ok$} node[midway,below]{$x_{p_1} = 10 \wedge x_\delim = 0$} (11,-3) to (fin); 
                \draw [color=green!60!black]  (C3) to node [midway,above] {$\ok$} node [midway,below] {$x_{p_3} = 5 \wedge x_\delim = 0$}(fin); 
                \draw [rounded corners, color=green!60!black] (C3) to (9,-5) to node [midway,above] {$\ok$} node[midway,below]{$x_{p_4} = 1 \wedge x_\delim = 0$} (11,-5) to (fin);

             \draw [rounded corners] (12) to (12.8,0) to (12.8,-2) to node [midway, above] {$\delim$} node [midway, below] {$x_{\delim} = 0$} (-1.3,-2) to (-1.3,-4) to (C1);
            \end{scope}
        \end{tikzpicture}
        \caption{The automaton $\A_{\varphi}$ corresponding to the formula 
        $\varphi = (p_1 \vee p_2 \vee \neg p_3)\land (\neg p_2 \vee \neg p_3 \vee p_4)\land (\neg p_1 \vee  p_3 \vee \neg p_4)$.
        The corresponding symbolic timed word is 
        $w_\varphi = ({\delim}, \top) (p_1, \top) (\delim, \top) (p_2, \top) (\delim, \top) (p_3, \top) (\delim, \top) (p_4, \top) (\delim, \top) ( \delim, \top) (\ok, \top)^3$. }
        \label{fig:np-hard-example}
\end{figure}
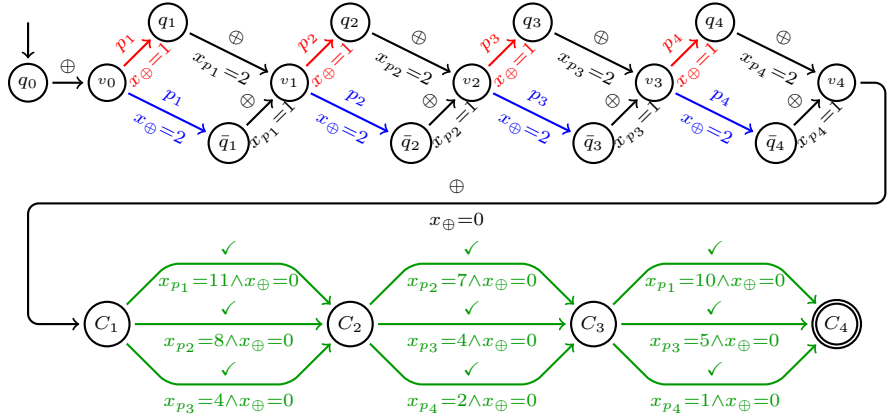
The transition relation  $\Delta$ is defined as follows.
For every $i\in \{1,2,\ldots,n\}$, from the state $v_{i-1}$, we keep two outgoing transitions, one to $q_i$ and the other to $\bar{q}_i$, which denotes the fact that the proposition $p_i$ has been set to $\tt$ and $\ff$, respectively. From each of $q_i$ and $\bar{q}_i$, there is a transition to $v_{i+1}$. The guards on these transitions ensure that $v_{i+1}$ is reached exactly after $3$ time units from $v_{i}$. Formally, 
$$v_{i-1} ~\red{\xrightarrow[\xs = 1]{p_i}}~ q_{i} \xrightarrow[x_{p_i} = 2]{\delim} v_i,\ \ 
v_{i-1} ~\blue{\xrightarrow[\xs = 2]{p_i}}~ \bar{q}_{i} \xrightarrow[x_{p_i} = 1]{\delim} v_i, \text{ for every } 1\le i \le n.$$

Notice that each run of $\Aa_{\varphi}$ can take only one of these paths for each $i$, since the guards on the transitions on $p_i$'s are disjoint.
A run of $\A_\varphi$, from $v_0$ to $v_n$, will determine the values of the propositions $p_i$'s: if the path through~$q_{i}$ is taken, we assign $p_i$ to $\tt$ and when the other path (through~$\bar{q}_{i}$) is taken, we assign $p_i$ to $\ff$.
Additionally, we add the following transition from the initial state $q_0$ that resets the clock $\xs$:
$q_0 \xrightarrow[\top]{\delim} v_{0}$.

The second phase begins from $C_1$. In this phase, we ensure that the automaton cannot elapse time at any of the states. We, for the sake of readability, keep a (dummy) transition between the two phases on the letter $\delim$, that resets the clock $\xs$, and we keep the guard in such a way that no time can elapse at $v_n$:
$v_n \xrightarrow[x_\delim = 0]{\delim} C_1$.
We then \emph{sequentially} check whether,
for every $h \in \{1,2,\ldots,m\}$, a literal in $C_h$ is made $\tt$ by the valuation chosen in the first phase.
This is achieved using appropriate guards on the transitions:
for every $1 \le j \le 3$, if $\ell_{h,j} = p_i$ for some $p_i$, then we add:
$C_{h} \xrightarrow[x_{p_i} = 3\times (n-i)+2 \wedge x_{\delim} = 0]{\ok}C_{h+1}$,
and if $\ell_{h,j} = \neg p_i$ for some $p_i$, then we add:
$C_{h} \xrightarrow[x_{p_i} = 3\times (n-i)+1 \wedge x_{\delim} = 0]{\ok}C_{h+1}$.
These transisitions are depicted in green in Figure~\ref{fig:np-hard-example}. However, each of them actually represent sets of transitions. For example, the transition $(C_1, \ok, x_{p_1} = 11 \wedge x_{\delim} = 0, C_2)$ denotes the set of transitions $\{(C_1, \ok, \psi, C_2) \mid \psi \in \bigcup_{d_2 \in \{7,8\}, d_3 \in \{4,5\}, d_4 \in \{1,2\}}\{x_{p_1} = 11 \wedge x_{p_2} = d_2 \wedge x_{p_3} = d_3 \wedge x_{p_4} = d_4\}\}$. 
Now note that, although the automaton in Figure~\ref{fig:np-hard-example} is non-deterministic as it is presented, if the green edges are replaced with the set of edges it represents as described above (removing redundant edges, if present), then the resulting automaton will indeed be deterministic. 
Finally, we make $C_{m+1}$ the unique final state of the automaton: $F = \{C_{m+1}\}$.

\smallskip
To illustrate the assignments to propositions, consider the following example. Suppose, $\ell_{h,2} = p_i$ for some $h$. Then, if a run of $\A_{\varphi}$ has visited the state $q_{i}$, it means $p_i$ has been set to $\tt$, and so is $\ell_{h,2}$. Note that this can be checked using the guard $x_{p_i} = 3\times(n-i)+ 2$, where the value of $x_{p_i}$ is the time elapsed since~$q_{i}$. On the other hand, if the run had visited the state $\bar{q}_{i}$, then $p_i$ would be $\ff$, and so is $\ell_{h,2}$, which can again be checked using the guard $x_{p_i} = 3\times(n-i)+1$. 

\medskip
\noindent\textbf{The symbolic timed word $w_{\varphi}$} is defined in a way so that $\sem{w_\varphi}$ is the set of all timed words whose untimed projection is $\delim\ 
  p_{1} \delim\  p_{2} \delim\  p_{n} \delim \ \delim \ \ok^{m}$. Note that, this untimed word is 
  visited in every accepting run in $\Aa_\varphi$.
  $$w_\varphi  := (\delim, \top)(p_1, \top)
  \ldots (\delim, \top) (p_n,\top)  
  (\delim,\top) \ (\delim, \top) \
  (\ok,\top) \ldots (\ok,\top).$$

The above reduction ensures the following equivalence.

\begin{restatable}{lemma}{hardness}
\label{lem:NP-h-correctness}
  Given a $3$-CNF formula $\varphi$, $\varphi$ is satisfiable iff $ \lang(\A_{\varphi}) \cap \sem{w_{\varphi}} \neq \emptyset$.
\end{restatable}

\begin{proof}
\label{proof:lem-hardness}
    Note that, there is an one-one correspondence between runs of $\A_\varphi$ until $v_{n}$ -- which we call a \emph{history} -- and valuations over $\{p_1, \cdots, p_{n}\}$: the valuation
$\iota_\pi$ associated with a history
$\pi = v_0 v'_1 v_1 v'_2 v_2 \cdots v_{n}$ is such that
$\iota_\pi(p_i) = 1$ if $v'_i = q_{i}$ and $\iota_\pi(p_i) = 0$ if
$v'_i = {\bar{q}_i}$. Indeed, this correspondence is a bijection: given a valuation $\iota$ over $\{p_1, \cdots, p_{n}\}$, there
is a unique history $\pi_\iota = v_0 v'_1 v_1 v'_2 v_2 \cdots v_{n}$
such that $v'_i = q_i$ if $\iota(p_i) = 1$, and $v'_i = {\bar{q}_i}$ if $\iota(p_i) = 0$.

 For a literal $\ell_{h,j}$, we define $var(\ell_{h,j}) = p_i$ if $\ell_{h,j} \in \{p_i, \neg p_i\}$ for some $1 \le i \le n$.
 
Now suppose that there is a valuation $\iota$ such that $\iota \models C_1 \wedge \dots \wedge C_m$.  Then, for every
  $1 \le h \le m$, there is $\ell_{h,j}$ (literal of $C_h$) such that  $\iota(\ell_{h,j})=1$. We show an accepting run of $\A_\varphi$ on $w_\varphi$. In the first phase, $\A_\varphi$ takes the run corresponding to $\pi_\iota$, and then in the second phase, from $C_h$, reading the letter $\ok$, it takes the transition corresponding to the variable $var(\ell_{h,j})$, where $\ell_{h,j}$ is $\tt$ in $C_h$. The construction of $\A_\varphi$ ensures that this transition on $\ok$ can indeed be taken from $C_h$. This is true for every $h \in \{1,2,\ldots,m\}$. This produces an accepting run of $\A_\varphi$ on $w_\varphi$.

  Conversely, suppose that there is an accepting run $\rho$ of $\A_\varphi$ on $w_\varphi$. The first part of the run, say $\pi$, generates a valuation in the first phase of the automaton $\iota_\pi$. Since $\rho$ is accepting, in the second phase, from every $C_h$, a guard on $var(\ell_{h,j})$ is satisfied for some $j \in \{1,2,3\}$. That implies, for every clause $C_h$, there is a literal $\ell_{h,j}$ which becomes true w.r.t. the valuation $\iota_\pi$. Hence, $\iota_\pi \models C_1 \wedge \dots \wedge C_m$. 
\end{proof}

This concludes the proof of Lemma~\ref{prop:lang-incl-np-hard}. 
\end{proof}

We can also show that the \textsc{Intersection Non-emptiness} problem for (possibly non-deterministic) ERA is in $\np$. Then the following theorem follows.

\begin{restatable}{theorem}{npcompleteness}
\label{corr:hardness-for-dera}
The \textsc{Intersection non-emptiness problem} is $\np$-complete, both for ERA and DERA.
\end{restatable}

\begin{proof}
\label{proof:NP-complete}
We have established hardness in Lemma~\ref{lem:NP-h-correctness}. Here we show that the \textsc{Intersection Non-emptiness} problem for ERA is in $\np$. 

    A certificate for this would be a path $\rho$ in $\Aa$, starting from the initial state and ending at an accepting state.
    We now show that, given a path $\rho$, it is possible to check deterministically in polynomial time if there exists a timed word $\word = (\sigma_1, t_1) (\sigma_2, t_2) \ldots (\sigma_n, t_n)$, such that $\Aa$ has an accepting run on $\word$ following the path $\rho$ and moreover, $\word \in \sem{w}$. 
    The following algorithm determines whether such a timed word  exists. 
    \begin{itemize}
        \item initialize $\varphi_\rho := \top$ and $\varphi_w := \top$
        \item for every edge $(q_i, \sigma_i, g_i, q_{i+1})$ in $\rho$:
        update $\varphi_\rho$ to $\varphi_\rho \wedge \tilde{g_i}$, where $\tilde{g_i}$ is the guard constructed from $g_i$ in the following manner --
        if $g_i$ contains an atomic constraint $x_{\sigma} \sim c$, then modify that constraint to $t_i-t_j \sim c$, where $j < i$ is the latest position in $w$ marked with the event $\sigma$.
        \item let $\varphi_w$ be the conjunction of all constraints coming from $w$: if the $i$-th letter present in $w$ is $(\sigma_i, g_i)$, then we similarly construct $\tilde{g_i}$ from $g_i$, and update $\varphi_w := \varphi_w \wedge \tilde{g_i}$.
        \item check if $\varphi_\rho \land \varphi_w$ is satisfiable where each $t_i \in \R_{\ge 0}$ and $t_i \le t_{i+1}$ for every $i$
    \end{itemize}
    Note now that the size $\varphi_w$ is that of $w$ and the size of $\varphi_\rho$ depends on the size of $\Aa$. Finally, checking if $\varphi_\rho \land \varphi_w$ is satisfiable or not is same as solving a linear program, which is known to be solvable in polynomial time. This proves that checking \textsc{Intersection Non-Emptiness} is in $\np$.
\end{proof}

However, if the sample set only contains region words and the ERA is deterministic, then the \textsc{Intersection non-emptiness problem} is in $\p$.
\begin{lemma}
\label{lem:reg-word-poly}
Given any DERA $\A$ and any symbolic timed word $w$ s.t. all guards in $\A$ and all guards in $w$ are $K$-simple constraints, for some $K \in \N$, then the \textsc{Intersection non-emptiness} problem can be solved in polynomial time.
\end{lemma}
\begin{proof}
 Without loss of any generality, we can assume that $\A$ is complete, \emph{i.e.}, for every state $q$, every letter $\sigma \in \Sigma$, and every $K$-simple constraint $r$, there exists exactly one outgoing transition from $q$ on $(\sigma,r)$. 
Now, since the guards present in $w$ are also $K$-simple constraints, one can just syntactically trace the path corresponding to $w$ in $\Aa$ (it always exists, since $\Aa$ is complete). 
 Then, it only remains to check if this path can indeed be taken by the automaton. This can be determined by checking whether $\sem{w} \neq \emptyset$ or not, which can be checked by solving a linear program.
\end{proof}

Note that, having to solve an $\np$-complete problem for every merge (unlike the case in RPNI for regular languages) is somewhat expected, as we deal with sequences of constraints over clocks. Even $\pspace$-completeness is common in analyzing Timed Automata, yet tools like UPPAAL~\cite{uppaal} remain useful in practice.

\subparagraph*{Complexity of $\leap$.}
The size of the $\PT$ constructed by $\leap$ is \emph{linear} in the size of $\splus$. The rest of the algorithm tries to merge states of this $\PT$ to get smaller models at each iteration, the total number of merges performed by the algorithm is \emph{polynomial} in the size of the $\PT$.
Now, in the most general case, when the sample contains symbolic timed words, at each iteration, the algorithm uses an $\np$ oracle (the algorithm for \textsc{Intersection non-emptiness problem}) to determine whether the merge is allowed or not. Therefore, the overall complexity of $\leap$ is $\p^{\np}$. 
Whereas, if the sample sets only contain region words, since \textsc{Intersection non-emptiness problem} is in $\p$ (Lemma~\ref{lem:reg-word-poly}), \leap\  also is in $\p$.
\begin{theorem}
\label{thm:complexity}
\leap\ executes in polynomial time when the sample sets contain only region words, and executes with a polynomial number of calls to an $\np$-oracle when the sample sets contain symbolic words in general.
\end{theorem}

\subparagraph*{Symbolic timed words vs region words.}
From any input sample set $S$ containing symbolic timed words, one can construct a sample set $S^{\rw}$ containing only region words 
(Lemma~\ref{lem:justify-region}),
and apply \leap\ on $S^{\rw}$. Since  \leap\ runs in polynomial time  in the size of $S^{\rw}$ (Lemma~\ref{lem:reg-word-poly}), and $S^{\rw}$ contains at most exponentially many samples than in $S$, the worst case complexities of the procedures are indeed similar. However, as also witnessed in the theory of TA, dealing with \emph{zones} is often more efficient than dealing with \emph{regions} in practice. We will see in Section~\ref{sec:experiments} that this is indeed also the case for \leap, \emph{i.e.}, handling symbolic timed words directly is often more efficient in practice.


\section{Completeness of \leap}
\label{sec:leap-completeness-sdera}

In this section, we will provide a completeness result for the  algorithm \leap. 
We will define the notion of \emph{characteristic sets} for ERA-recognizable languages. Intuitively, these are sample sets that contain necessary information for the learning algorithm to infer the correct language. 
We show that, given an ERA-recognizable language $L$, there exist characteristic sets $S_L$ such that, when $S_L$ is given as input to \leap, it will return an ERA recognizing $L$. 

In our construction, the characteristic sets  only contain region words.
In the worst case, the size of these sets can be \emph{exponential} in the size of the alphabet $\Sigma$ and also in the size of a minimal DERA $\A$ recognizing $L$. However, one can show that the exponential blowup in the size of the characteristic sets (containing only region words) is necessary. 
The reason being the class of ERA is not polynomially learnable, as shown in Theorem $1$ of~\cite{VWW08}. Even though the original result in~\cite{VWW08} was shown for the class of DTA, the same proof also holds for the class of ERA, by relabeling the transitions of the TA in Figure~$2$ of~\cite{VWW08} to form an ERA.

\begin{lemma}[\cite{VWW08}]
\label{lem:ERA-not-poly-recognizable}
    The class of ERA is not polynomially learnable.
\end{lemma}

Now, notice that, thanks to Theorem~\ref{thm:complexity}, if a sample set contains only region words, then \leap\ runs in polynomial time in the size of such a set. 
Therefore, if for every ERA-recognizable language $L$ there would exist a polynomial size characteristic set containing only region words, then that would contradict Lemma~\ref{lem:ERA-not-poly-recognizable}.  Thus, we conclude the following.
\begin{corollary}
\label{cor:no-poly-ch-set}
    There exist ERA-recognizable languages for which there cannot exist polynomial size characteristic sets for \leap\ containing only region words. 
\end{corollary}

However, for every ERA-recognizable language $L$, one can construct  characteristic sets  containing only region words, for which \leap\ returns a DERA recognizing $L$. This is formalized in the following theorem. 

\begin{restatable}[]{theorem}{corrleap}
\label{th:completeness}
    For every $K$-ERA recognizable language $L$, there exist characteristic sets $S$, such that, when provided as input to \leap, it returns a DERA $\Aa$ with $\lang(\Aa) = L$. 
\end{restatable}

The rest of this section is dedicated to prove the above theorem, in which we will
assume that the alphabet $\Sigma$ and the parameter $K$, the maximal constant that appears in the guards of a DERA recognizing a timed language $L$, are known. This is a reasonable assumption, and related research in active learning frameworks also assume this~\cite{GJL10,Waga23,ATVA2024}.
In the following, we will show a construction of characteristic sets, that \emph{only contain region words}, for any  ERA-recognizable language $L$ over $(\Sigma, K)$. The following lemma shows that this assumption is not at the cost of any generality.

\begin{lemma}
\label{lem:justify-region}
Let $w$ be a symbolic timed word $(\sigma_1,g_1)(\sigma_2,g_2)\dots (\sigma_n,g_n)$ over $(\Sigma, K)$.
Then there exists a finite set of region words $(w_i=(\sigma_1,r^i_1)(\sigma_2,r^i_2)\dots (\sigma_n,r^i_n))_{1 \leq i \leq n}$ over $(\Sigma, K)$ such that $\sem{w}=\bigcup_{1 \leq i \leq n} \sem{w_i}$.
\end{lemma}

The construction we provide of characteristic sets from an ERA recognizable language $L$ is similar to the construction for regular languages for RPNI~\cite{oncina1992}. The proof of the fact that \leap\ indeed returns an automaton with the correct language when provided with a characteristic set is, however, novel and non-trivial. 
The proof of completeness for RPNI can also be proved alternatively using similar lines of arguments as ours.

Given a region word $u$ and a timed language $L$, we say $u \in L$ if $\sem{u} \subseteq L$. Then, note that, for every region word $u$, if $\sem{u} = \emptyset$, then for every timed language $L$, $u \in L$. We will call such a word, an \emph{empty} word.

\begin{definition}[tail]
    \label{def:tail}
    Let $u \in {\sf RW}(\Sigma,K)$ be a region word and $L$ be a timed language recognizable by a $K$-ERA over $\Sigma$. Then, the tail of $u$ is a set of region words defined as:
    $\tail{u} = \{w \in {\sf RW}(\Sigma,K) \mid u.w \in L\}$.
\end{definition}

\begin{lemma}
\label{lem:finitely-many-tails}
    For every $K$-ERA recognizable language $L$, the set $\{\tail{u} \mid u \in \RW{\Sigma, K}\}$ is finite, and in the worst case, is of exponential size w.r.t. the size of a minimal DERA recognizing $L$.
\end{lemma}

\begin{proof}
    Let $L$ be a $K$-ERA recognizable language and let $\Aa_{era}$ be an ERA with $\lang(\Aa_{era}) = L$. 
    Now, a \emph{simple DERA}~\cite{GJL10} (SDERA, for short) is a DERA, where all the guards are simple constraints and every path leading to an accepting state is satisfiable, further, w.l.o.g. we can assume that for every region word the SDERA contains a path on this word. For every ERA, there exists an SDERA that enjoys the following property (thanks to Lemma~19 of~\cite{GJL10}): for every pair of region words $u, v$, if both $u$ and $v$ lead to the same state then for every region word $w$, $\sem{u.w} \neq \emptyset$ iff $\sem{v.w} \neq \emptyset$.
    \footnote{Intuitively, this happens because this SDERA closely corresponds to the region automaton.} Let $\Aa$ be the minimal SDERA corresponding to $\Aa_{era}$. 
    
    Let $u, v$ be two region words that lead to the same state, say $q$, in $\Aa$. Let $w \in \tail{u}$. Now, if $\sem{u.w} = \emptyset$, then due to the (above-mentioned) property of $\Aa$, we have that $\sem{v.w} = \emptyset$, and hence $w \in \tail{v}$. Now, suppose $\sem{u.w} \neq \emptyset$, then $\sem{v.w} \neq \emptyset$. Since $\lang(\Aa) = L$ and $w \in \tail{u}$, we can deduce that 
    there exists a path in $\Aa$ from $q$ on the word $w$,  
    that leads to an accepting state in $\Aa$. Since $v$ also leads to the same state as $u$,  we get that $\sem{v.w} \subseteq L$, and hence $w \in \tail{v}$. This shows that $\tail{u} \subseteq \tail{v}$, we can similarly argue that $\tail{v} \subseteq \tail{u}$.
    Therefore, we get that, for every two region words that lead to the same state in $\Aa$ have the same tail language. This implies that the set of all tail languages is finite, and since the size of a minimal SDERA corresponding to $L$ can be of size exponential w.r.t. the size of a minimal DERA recognizing $L$, the size of this set can be of exponential size.
\end{proof}

For every $K$-ERA recognizable timed language $L$, we define 
$\mathsf{Prefix}(L) := \{u \in  \RW{\Sigma,K} \mid \exists w \in \RW{\Sigma,K}: u.w \in L\}$ and 
$\mathsf{Suffix}(L) := \{w \in  \RW{\Sigma,K}\mid \exists u \in  \RW{\Sigma,K}: u.w \in L\}$.
In the following, we will use the total order $\order$ defined between symbolic timed words that was mentioned in Section~\ref{order-symbolic-words}.

\begin{definition}
    \label{def:sp-sdera}
    Let $L$ be a language recognizable by a $K$-ERA. 
    The set of shortest prefixes of $L$ is: 
    $\sp{L} = \{u \in \RW{\Sigma,K} \mid \nexists v \in \RW{\Sigma,K}: v  \order u : \tail{u} = \tail{v}\}$.
    The kernel of $L$ is:
    $\ker{L} = \{(\varepsilon, \mathbf{0})\} \cup \{ u. (\sigma, r) \in \RW{\Sigma,K} \mid u \in \sp{L} \text{ and } u. (\sigma, r) \in \mathsf{Prefix}(L) \}$.
\end{definition}

We now define \emph{characteristic sets} corresponding to $K$-ERA recognizable languages.

\begin{definition}[characteristic set]
\label{def:char-set-sdera}
A sample set  $S = (\splus, \sminus)$ containing only region words, is \emph{characteristic} for a $K$-ERA recognizable language $L$ if -- 
    (i) $\forall v \in \ker{L}, \textit{ if }  
    v \in L \textit{ then } v \in \splus;
    \textit{ else } \exists w \in \mathsf{Suffix}(L) \text{ where } vw \in \splus \cap L$; and
    (ii)
    $\forall u \in \sp{L}, \forall v \in \ker{L}, 
    \textit{ if } \tail{u} \neq \tail{v} 
    \textit{ then } \exists w \in \mathsf{Suffix}(L)$ 
    such that either (i)
    $u.w \in \splus \cap L\ \& \ v.w \in \sminus\cap \overline{L} \text{ or (ii) }
    u.w \in \sminus\cap \overline{L} \ \& \ v.w \in \splus\cap L$.
\end{definition}

In the definition above, $\overline{L}$ denotes the complement of $L$. Now,
the following result directly follows from Lemma~\ref{lem:finitely-many-tails}.
\begin{corollary}
    For every $K$-ERA recognizable language $L$, there exist characteristic sets that are finite.
\end{corollary}

\begin{lemma}
    \label{lem:same-tail-same-extension}
    Let $u, v \in \RW{\Sigma,K}$, $(\sigma, r) \in \Sigma \times \Region$ and let $L$ be a timed language recognizable by a $K$-ERA. Then, $\tail{u} = \tail{v}$ implies $\tail{u.(\sigma, r)} = \tail{v. (\sigma, r)}$.
\end{lemma}

\begin{proof}
    Let $w \in \tail{u.(\sigma, r)}$. We then have the following:
    \begin{align*}
        & u.(\sigma, r).w \in L ~(\text{from Definition~\ref{def:tail}})\\
        \iff & (\sigma, r). w \in \tail{u} \\
        \iff & (\sigma, r). w \in \tail{v} ~(\text{since } \tail{u} = \tail{v}) \\
        \iff & v. (\sigma, r) . w \in L ~(\text{from Definition~\ref{def:tail}}) \\
        \iff & w \in \tail{v.(\sigma, r)}
    \end{align*}
    Therefore, we get that $\tail{u.(\sigma, r)}  = \tail{v.(\sigma, r)}$.
\end{proof}

Here we note few remarks on the definitions of $\sp{L}$, $\ker{L}$.

\begin{lemma}
    \label{lem:sp-subset-nl}
    $\sp{L} \subseteq \ker{L}$.
\end{lemma}

\begin{proof}
    $(\varepsilon,\mathbf{0})$ is in both $\sp{L}$ and $\ker{L}$.
    Now, choose a word $w \in \sp{L}$ which is not $(\varepsilon, \mathbf{0})$. 
    The claim is that $w$ must be of the form $w' . (\sigma, r)$, where $w' \in \sp{L}$. 
    Suppose, this is not the case, that is, $w' \notin \sp{L}$. Then, there exists $w'' \order w'$ with $\tail{w''} = \tail{w'}$. 
    This implies, $\tail{w''.(\sigma, r)} = \tail{w'. (\sigma, r)}$ and also $w''.(\sigma, r) \order w'.(\sigma, r)$. However, this contradicts the fact that $w = w'.(\sigma, r) \in \sp{L}$. Therefore, $w' \in \sp{L}$ and hence $w \in \ker{L}$. 
\end{proof}

\begin{remark}
    There may exist two simple words $u \in \sp{L}$ and $v \in \ker{L}$ with $\tail{u} \neq \tail{v}$ for which the only possibility is that 
    there exists a simple word $w$ such that $\sem{u.w} = \emptyset$ and $v.w \notin L$. In this case, we would have to add $u.w$ to $\splus$ and $v.w$ to $\sminus$. Therefore, the set $\splus$ may contain empty words. The set $\sminus$, however, can never contain an empty word.
\end{remark}

In the subsequent results to be presented in this section, 
one iteration of \leap\ (c.f. Algorithm~\ref{alg:LEAP}) refers to either of the following: (i) a blue node has been merged with a red node and subsequently a few pairs (possibly none) of states have been folded -- this implies every automaton computed by \leap\ after each of its iterations, is deterministic (since we are working with region words) -- or (ii) a blue node has been promoted to a red node.

    \label{rem:two-new-edges-for-merge}
    In an iteration where a merge and several folds occur, there are two sets of transitions that get added: 
    (i) a new edge between two red states, and (ii) a (possibly empty) set of  edges from a set of red states to a set of blue states.
 
We now introduce a notation and establish some results that will help us to prove Lemma~\ref{lem:leap-invariants}.

Given a $K$-DERA $\Aa = (Q, q_{in}, \Sigma, \delta, F)$, we define $dfa(\Aa)$ (a similar notion has been used in literature, for example in~\cite{GJL10}) to be the DFA $(Q, q_{in}, \Sigma \times \SimC{\Sigma,K}, \delta', F)$, where the alphabet is the set of event-guard pairs appearing in $\Aa$ and $\delta' :  Q \times (\Sigma \times \SimC{\Sigma,K}) \to Q$ is defined as: $\delta'(q, (\sigma, r)) := \delta(q, \sigma, r)$ if $\delta(q, \sigma, r)$ is defined, and undefined otherwise. 
We will use $\lang(dfa(\Aa))$ to denote the set of all words in $(\Sigma \times \SimC{\Sigma,K})^*$ that are accepted by $dfa(\Aa)$. 

Now, note that, at every iteration of \leap, 
no path gets removed (it only gets rerouted) during an iteration. 
The implies the following result.

\begin{lemma}
    \label{lem:dfa-language-preserved-leap}
    For every iteration $i \ge 1$ of \leap, $\lang(dfa(\Aa_{i-1})) \subseteq \lang(dfa(\Aa_i))$ and therefore, we also have $\lang(\Aa_{i-1}) \subseteq \lang(\Aa_i)$.
\end{lemma}

\begin{lemma}
    \label{lem:leap-sp-nl-same-state-same-tail}
    Let $\Aa_i$ be the automaton computed by \leap\ after $i$ iterations and let $u \in \sp{L}$ and $v \in \ker{L}$.
    Then, $\delta_i(q_{in}, u) = \delta_i(q_{in}, v)$ implies $\tail{u} = \tail{v}$.
\end{lemma}

\begin{proof}
    Assume $\tail{u} \neq \tail{v}$, then from the construction of characteristic sets (Definition~\ref{def:char-set-sdera}), we know that there exists a region word $w \in \mathsf{Suffix}(L)$ such that either (i) $u.w \in \splus$ and $v.w \in \sminus$ or (ii) $u.w \in \sminus$ and $v.w \in \splus$.
    Let us consider the first case. 
    Note that, firstly $\splus = \lang(dfa(\Aa_0))$ and (from Lemma~\ref{lem:dfa-language-preserved-leap}) for every iteration $j$ of \leap, $\lang(dfa(\Aa_j)) \subseteq \lang(dfa(A_{j+1}))$. Therefore, we know $\splus \subseteq \lang(dfa(\Aa_i))$ and hence $u.w \in \lang(dfa(\Aa_i))$. 
    Let $q_u$ be the state such that $\delta_i(q_{in}, u) = q_u$. Since, $u.w \in \lang(dfa(\Aa_i))$, we know that there exists a path from $q_u$ to an accepting state on the word $w$.
    Now, since $u$ and $v$ lead to the same state in $\Aa_i$, $v.w \in \lang(dfa(\Aa_i))$ and therefore $v.w \in \lang(\Aa_i)$. However, this contradicts the fact that $\lang(\Aa_i) \cap \sminus = \emptyset$. The other case can be argued similarly.
\end{proof}

Let $\Aa_i= (Q_i, q_{in}, \Sigma,  \delta_i, F_i)$ be the $K$-ERA computed by \leap\ after $i$ iterations.
We define the shortest prefix to a state $q$ in  $\Aa_i$ as the word $\ssp{i}{q} := \min\{w\in \RW{\Sigma,K} \mid \delta_i(q_{in}, w) = q \}$.

\begin{lemma}
\label{lem:leap-sp-dont-change}
Let $\Aa_i = (Q_i, q_{in}, \Sigma,  \delta_i, F_i)$ be the automaton computed by \leap\ after $i$ iterations. Then
    for every red state $q \in Q_i$, $\ssp{i}{q} = \ssp{i+1}{q}$.
\end{lemma}

\begin{proof}
Let $q \in Q_i$ be a red state.
First notice that,
for every iteration $j$ of \leap, every path present in $\Aa_j$ remains in $\Aa_{j+1}$.
Hence, $\ssp{i+1}{q} \ordereq \ssp{i}{q}$. Now assume that $u := \ssp{i+1}{q} \order \ssp{i}{q}$. 
Since $q$ is already a red state in $\Aa_i$, $\qu$ must have been processed before the $i$-th iteration of \leap\ and has been merged, or been promoted to some $q'$ that is red in $Q_{i}$. Hence, $\delta_i(q_{in}, u) = q'$.
Since every path is preserved during an iteration of \leap, it must be that $\delta_{i+1}(q_{in}, u) = q'$.
Since we have assumed that $u = \ssp{i+1}{q}$, it must be that $\delta_{i+1}(q_{in}, u) = q$. Now, since $\Aa_{i+1}$ is deterministic, we get that $q = q'$.
This implies $\delta_i(q_0, u) = q$, which contradicts the fact that $u \order \ssp{i}{q}$.
\end{proof}

We now prove a set of invariants maintained by \leap. This lemma will help us prove Theorem~\ref{th:completeness}.

\begin{lemma}
\label{lem:leap-invariants}
    Let $\Aa_i = (Q_i, q_{in}, \Sigma, \delta_i, F_i)$ be the $K$-DERA computed by \leap\ after $i$ iterations. Then, the following results hold:
    \begin{enumerate}
        \item 
        \label{IH:sp_i_in_spL}
        for every red state $q \in Q_i$, $\ssp{i}{q} \in \sp{L}$,
        \item 
        \label{IH:same_state_same_tail}
        for every $u, v \in \RW{\Sigma,K}$, $\delta_i(q_{in}, u) = \delta_i(q_{in}, v) \implies \tail{u} = \tail{v}$,
        \item
        \label{IH:intermediary}
        for every red state $q$ and every word $y$ if $q \xlongrightarrow{y}q_f$ is a run in  $\A_i$ that does not visit any red state and if this run is not present in $\Aa_{i-1}$, then $y \in \tail{\ssp{i}{q}}$,
        \item
        \label{IH:L_i_in_L}
        $\lang(\Aa_i) \subseteq L$, and
        \item
        \label{IH:same_tail_merge}
        for every pair of words $u, v \in \RW{\Sigma,K}$ such that $q_u$ is a red state and $q_v$ is a non-red state in~$\A_i$, if $q_u = \delta_i(q_{in}, u) \neq \delta_i(q_{in}, v) = q_v$ and $\tail{u} = \tail{v}$, then the states $q_u$ and $q_v$ can be merged, \emph{i.e.}, $\lang(\Aa_i^{q_u, q_v}) \cap \sminus = \emptyset$.
    \end{enumerate}
\end{lemma}

\begin{proof}
\begin{description}
    \item[(Base case)] 
        1. In $\Aa_0$, that is in the $\PT$, the only red state is $q_{in}$ and $\ssp{0}{q_{in}} = (\varepsilon, \mathbf{0}) \in \sp{L}$.
        \item 
        
        2. Since $\Aa_0$ is a tree, there are no two distinct words $u,v$ such that $\delta_0(q_{in}, u) = \delta_0(q_{in}, v)$, therefore, the result holds trivially.
        \item 
        
        3. The result holds trivially, since $A_0$ is the automaton in the first iteration.
        \item 
        
        4. For every $w \in \lang(\Aa_0)$, either $w$ is an empty word (in this case, $w \in L$), or $w \in \lang(dfa(\Aa_0))$. Since, $\lang(dfa(\Aa_0)) = \splus \subseteq L$, the result holds.
        \item 
        
        5. Suppose, $w \in \lang(\Aa_0^{q_u, q_v}) \cap \sminus$. 
        Without any loss of generality, let us assume that $q_u$ is the initial state and $q_v$ is blue. Since $q_u$ is the initial state $\Aa_0$, $u = (\varepsilon, \mathbf{0})$.
        Now, $w$ must be of the form $v^*.w_u$, where either (i) $w_u \in \splus$ or (ii) $v.w_u \in \splus$.
        In the first case, $w_u \in L$. 
        In the second case, $v.w_u \in L$, and since $u = (\varepsilon, \mathbf{0})$ and $\tail{u} = \tail{v}$, we can deduce that $(\varepsilon, \mathbf{0}).w_u \in L$ as well. 
        Since, $ \tail{(\varepsilon, \mathbf{0})} = L$, we get that in both cases, $w_u \in \tail{(\varepsilon,\mathbf{0})}$. 
        
        Now, since $\tail{(\varepsilon, \mathbf{0})} = \tail{v}$ we can argue inductively (by suffixing both sides with $v$) that $\tail{(\varepsilon, \mathbf{0})} = \tail{v^*}$.
        Finally, since $w_u \in \tail{(\varepsilon, \mathbf{0})}$, we can deduce that $w_u \in \tail{v^*}$ as well. Therefore, $v^*. w_u \in L$, this contradicts the fact that $v^*. w_u \in \sminus$.
        
    \item[(Induction hypothesis)] We assume that all the results hold for every iteration $j < i$.

    \item[(Induction step)] 
        1. If a merge (and subsequent folds) has occurred at the $i$-th step of \leap, since no state has been added to the red part of $\A_{i-1}$, from Lemma~\ref{lem:leap-sp-dont-change}, we get that $\ssp{i-1}{q} = \ssp{i}{q}$ for all red states $q$, and the statement follows from the outer Induction Hypothesis~(\ref{IH:sp_i_in_spL}).

        \smallskip
        Now assume that the $i$-th step of \leap\ was a promotion -- some $q_u$ has been promoted to red. Notice that for every red state $q$ other than $q_u$ in $\Aa_i$, since it was also a red state in $\Aa_{i-1}$, using Lemma~\ref{lem:leap-sp-dont-change}, and outer Induction Hypothesis~(\ref{IH:sp_i_in_spL}), $\ssp{i}{q} \in \sp{L}$. Now, towards a contradiction, let $(\tilde{u} :=)~ \ssp{i}{q_u} \notin \sp{L}$. Then, from the definition of $\sp{L}$ (Definition~\ref{def:sp-sdera}) it follows that there exists $w \order \tilde{u}~ (= \ssp{i}{q_u})$ such that $\tail{w} = \tail{\tilde{u}}$. Since $w \in \sp{L}$ and $\sp{L} \subseteq \ker{L}$ (Lemma~\ref{lem:sp-subset-nl}), it follows that $w \in \ker{L}$ as well. 
        Then, from the construction of characteristic sets, we know that $\splus$ contains a word $w.w'$, for some $w' \in \RW{\Sigma,K}$.
        Therefore, the $\PT$ (\emph{i.e.}, $\Aa_0$) contains a state $q_w$ where $\delta_0(q_{in}, w) = q_w$. Since $w \order \tilde{u}$ and $q_u$ is already red in $\Aa_i$, $q_w$ must have become red (either by promotion or a merge/fold) at an iteration $j < i$. Further, since $\ssp{i}{q_w} = w \order \tilde{u} = \ssp{i}{q_u}$, we know that $q_u \neq q_w$ in $\Aa_i$ and therefore, also in $\Aa_{i-1}$.
        Then, from the induction hypothesis~(\ref{IH:same_tail_merge}) we know that $\lang(\Aa_{i-1}^{q_u, q_w}) \cap \sminus = \emptyset$. This contradicts the fact that, \leap\ did not merge $q_u$ with $q_w$. 
        \item 
        
        2. This result follows directly from the following intermediate lemma.
        \medskip
        \begin{lemma}
            \label{lem:same-state-same-tail-ssp}
            Let $\Aa_i$ be the automaton obtained after $i$ iterations of \leap. Then, for every $t \in \RW{\Sigma,K}$ and for every red state $q \in Q_i$, if $\delta_i(q_{in}, t) = q$, then $\tail{t} = \tail{\ssp{i}{q}}$.
        \end{lemma}
        \begin{proof}
            Let $q_u$ (red) and $q_v$ (blue) be the two states in $\Aa_{i-1}$ that have been merged at the $i$-th iteration.
            Assume that in $\Aa_{i-1}$, there were the following: $\delta_{i-1}(q_{in}, w) = q_w$, $\delta_{i-1}(q_w, (\sigma, r)) = q_v$ and $\delta_{i-1}(q_{in}, u) = q_u$; then after the merge, $\Aa_i$ contains the following: $\delta_i(q_{in}, w) = q_w$,  $\delta_i(q_{in}, u) = q_u$ and the new transition $\zeta := (q_w, \sigma, r, q_u)$. 
            
            Let $\rho$ be the run of $\Aa_i$ on $u$ terminating at $q$. Since $q$ is a red state, $\rho$ traverses only through the red edges present in $\Aa_i$.
            Note that, (as remarked in  Page~\pageref{rem:two-new-edges-for-merge}) if $\rho$ does not contain the transition $\zeta$, then $\rho$ was already present in $\Aa_{i-1}$, \emph{i.e.}, $\delta_{i-1}(q_{in}, t) = q$ and then the result follows from induction hypothesis~(\ref{IH:same_state_same_tail}).

            Now let us consider the case when, $\rho$ has at least one occurrence of $\zeta$. We prove the result by induction on the number of occurrences ($k$) of $\zeta$ in $\rho$.
        Note that, if $\rho$ has $k$ occurrences of $\zeta$, without loss of any generality, $\rho$ can be written as:
        \[
        \rho =
        q_{in} \xlongrightarrow[]{t_0} 
        \blue{q_w \xrightarrow[r]{\sigma}
        q_u }\xlongrightarrow[]{t_1} 
        \blue{q_w \xrightarrow[r]{\sigma}
        q_u}
        \cdots 
        \blue{q_w \xrightarrow[r]{\sigma}
        q_u} \xlongrightarrow[]{t_k}q
        \]
        where, $t = t_0 . (\sigma, r) . t_1 \ldots t_{k-1} . (\sigma, r) . t_k$ and for every $0 \le i \le k$: (i) $t_i \in \RW{\Sigma,K}$, and (ii) every transition in $t_i$ is in $A_{i-1}$. 

        \textbf{Base case ($k = 1$):} In this case, $\rho$ is 
        $q_{in} \xlongrightarrow[]{t_0} 
        \blue{q_w \xrightarrow[r]{\sigma}
        q_u }\xlongrightarrow[]{t_1} q$. 
        Since $q_{in} \xrightarrow{t_0} q_w$ was also present in $\Aa_{i-1}$, using induction hypothesis~(\ref{IH:same_state_same_tail}) we get: 
        \setcounter{equation}{0}
        \begin{align}
            \tail{t_0} &= \tail{\ssp{i-1}{q_w}} \nonumber \\
            \implies \tail{t_0}&= \tail{\ssp{i}{q_w}}~(\text{using Lemma~\ref{lem:leap-sp-dont-change}}) 
        \end{align}
        Now, $\ssp{i}{q_w}.(\sigma, r) \in \ker{L}$ and $\ssp{i}{q_u} \in \sp{L}$, then using Lemma~\ref{lem:leap-sp-nl-same-state-same-tail} we get: 
        \begin{align}
            \tail{\ssp{i}{q_u}} = \tail{\ssp{i}{q_w}.(\sigma, r)}
        \end{align}
        Since $q$ is a red state, $q_u \xrightarrow{t_1} q$ was also present in $\Aa_{i-1}$ and, further, $\delta_{i-1}(q_{in}, \ssp{i-1}{q_u}) = q_u$, we can conclude using induction hypothesis~(\ref{IH:same_state_same_tail}): 
        \begin{align*}
            &\tail{\ssp{i-1}{q}} = \tail{\ssp{i-1}{q_u}. t_1} \\
            \implies &\tail{\ssp{i}{q}} = \tail{\ssp{i}{q_u}.t_1} ~(\text{using Lemma~\ref{lem:leap-sp-dont-change}}) \\
            \implies & \tail{\ssp{i}{q}} = \tail{\ssp{i}{q_w}.(\sigma, r).t_1} ~(\text{using equation~(2)}) \\
            \implies & \tail{\ssp{i}{q}} = \tail{t_0. (\sigma, r). t_1} ~(\text{using equation~(1)}) \\
            \implies & \tail{\ssp{i}{q}} = \tail{t}
        \end{align*}
        \textbf{Induction hypothesis.} We assume that the result holds when the number of occurrences of $\zeta$ is less than $k$.
        \medskip
        
        \textbf{Induction step.} Assume that $\rho$ contains $k$ occurrences of $\zeta$, then $\rho$ is:
        \[
        \rho :=
        q_{in} \xlongrightarrow[]{t_0} 
        \blue{q_w \xrightarrow[r]{\sigma}
        q_u }\xlongrightarrow[]{t_1} 
        \blue{q_w \xrightarrow[r]{\sigma}
        q_u}
        \cdots
        \blue{q_w \xrightarrow[r]{\sigma}
        q_u} \xlongrightarrow[]{t_{k}}q
        \]
        Consider $\rho'$ where the transitions $q_{in} \xrightarrow{t_0} q_w \xrightarrow[r]{\sigma} q_u$ are replaced with $q_{in} \xrightarrow{\ssp{i}{q_u}} q_u$:
        \[
        \rho' :=
        q_{in} \xrightarrow[]{\ssp{i}{q_u}}
        q_u\xlongrightarrow[]{t_1} 
        \blue{q_w \xrightarrow[r]{\sigma}
        q_u}
        \cdots
        \blue{q_w \xrightarrow[r]{\sigma}
        q_u} \xlongrightarrow[]{t_{k}}q
        \]
        Note that, since $q_u$ was also a red state in $\Aa_{i-1}$, from Lemma~\ref{lem:leap-sp-dont-change} we can deduce that $\ssp{i}{q_u}$ does not include the transition $\zeta$. 
        Thus, the number of occurrences of $\zeta$ in $\rho'$ is $k-1$. 
        Using the (inner) induction hypothesis we then get that:
        \begin{align}
            \tail{\ssp{i}{q_u}.t_1.(\sigma, r) \ldots (\sigma, r) . t_k} = \tail{\ssp{i}{q}}
        \end{align}
        Using the same argument as used in the base case, we get equations~(1) and~(2).
        Then,
        \begin{align*}
            &\tail{t_0} = \tail{\ssp{i}{q_w}} ~(\text{equation~(1)})  \\
           \implies &\tail{t_0.(\sigma, r)} = \tail{\ssp{i}{q_w}.(\sigma, r)} \\
           \implies & \tail{t_0. (\sigma, r)} = \tail{\ssp{i}{q_u}} ~(\text{using equation~(2)}) \\
           \implies & \tail{t_0. (\sigma, r).t_1.(\sigma,g)\ldots(\sigma, r) . t_k} = \tail{\ssp{i}{q_u}.t_1.(\sigma, r)\ldots (\sigma, r).t_k} \\
           \implies & \tail{t_0. (\sigma, r).t_1.(\sigma,g)\ldots(\sigma, r) . t_k} = \tail{\ssp{i}{q}} ~(\text{using equation~(3)})
        \end{align*}
        Therefore, we get that $\tail{t} = \tail{\ssp{i}{q}}$.
        \end{proof}
        Now, let $u$ and $v$ be two simple words such that $\delta_i(q_{in}, u) = \delta_i(q_{in}, v) = q$ (say).
        If $q$ is a red state, then from Lemma~\ref{lem:same-state-same-tail-ssp} we can deduce that $\tail{u} = \tail{\ssp{i}{q}} = \tail{v}$. 
        On the other hand, if $q$ is not a red state, then we know there exists a unique red state $q_1$ from which $q$ can be reached via following a path, on some simple word $t$, that does not visit any red states. Then, $\ssp{i}{q} = \ssp{i}{q_1}.t$. Since the two words $u$ and $v$ both lead to $q$, it must be the case that $t$ is a suffix of both $u$ and $v$, i.e. $u = u_0. t$ and $v = v_0 . t$, where $\delta_i(q_{in}, u_0) = q_1 = \delta_i(q_{in}, v_0)$. Then, again, from Lemma~\ref{lem:same-state-same-tail-ssp} we get that $\tail{u_0} = \tail{\ssp{i}{q_1}} = \tail{v_0}$ and adding $t$ to the suffix, we get that $\tail{u_0. t} = \tail{\ssp{i}{q_1}.t} = \tail{v_0. t}$, that is, $\tail{u} = \tail{\ssp{i}{q}} = \tail{v}$.  
        
        \item
        3. Since the path $q \xlongrightarrow{y} q_f$ does not visit any red states, and since the algorithm \leap\ does not modify such paths, 
        there was a path $q' \xlongrightarrow{y} q_f$ in $\Aa_{i-1}$ and $q,q'$ have been merged/folded during the $i$-th step. 
        From the induction hypothesis~(\ref{IH:L_i_in_L}), we get that $\ssp{i-1}{q'}.y \in L$ and hence $y \in \tail{\ssp{i-1}{q'}}$. Now, since $\delta_i(q_{in}, \ssp{i-1}{q'}) = q = \delta_i(q_{in}, \ssp{i}{q})$, from the induction step~(\ref{IH:same_state_same_tail}) we get that $\tail{\ssp{i-1}{q'}} = \tail{\ssp{i}{q}}$ and therefore $y \in \tail{\ssp{i}{q}}$.
        
        \item 
        
        4. Let $t \in L(\A_{i})$.
        If the $i$-th step of \leap\ was a promotion, then $L(\A_{i-1}) = L(\A_i)$ and the statement holds from outer induction hypothesis~(\ref{IH:L_i_in_L}).
        Now, let us consider the case when at the $i$-th iteration, 
        two states have been merged and subsequently a few pairs of states (possibly none) have been folded.  Consider the run $\rho := q_{in} \xlongrightarrow{t}q_f$ of $\A_i$ on the word $t$. There are now two possibilities:
            \item 
            
            (i). $\rho$ only visits red states. Now, either (i) $q_f$ was also an accepting state in $\Aa_{i-1}$, or (ii) $q_f$ was not an accepting state in $\Aa_{i-1}$, in which case, an accepting state $q_f'$ in $\Aa_{i-1}$ has been merged/folded with $q_f$ during the $i$-th iteration. 
            In the first case, we have that $\ssp{i-1}{q_f} \in L$ (since from the induction hypothesis~(\ref{IH:L_i_in_L}), we know that $\lang(\Aa_{i-1}) \subseteq L$). Now, since $\delta_i(q_{in}, \ssp{i-1}{q_f}) = q_f = \delta_i(q_{in}, t)$, we get that $\tail{\ssp{i-1}{q_f}} = \tail{t}$ and therefore, $t \in L$.
            In the second case, we have $\ssp{i-1}{q_f'} \in L$. Then, following similar arguments as in the first case, we can deduce that $\tail{\ssp{i-1}{q_f'}} = \tail{t}$. We can then conclude that  $t \in L$ as well.
            
            \item 
            
            (ii). $q_f$ is not a red state. Then, let $q$ be the last red state present in $\rho$, \emph{i.e.}, $\rho = q_{in} \xlongrightarrow{t_1}q\xlongrightarrow{t_2}q_f$, for some $t_1, t_2 \in (\Sigma\times \Region)^*$, where $q \xrightarrow{t_2} q_f$ does not visit any red states. Then from induction hypothesis~(\ref{IH:intermediary}), we get $t_2 \in \tail{\ssp{i}{q}}$, and further, from induction hypothesis~(\ref{IH:same_state_same_tail}), we get that $\tail{\ssp{i}{q}} = \tail{t_1}$. Therefore, $t_2 \in \tail{t_1}$ and hence, $t_1.t_2\in L$, \emph{i.e.}, $t\in L$.

        \item 
        
        5. Let $u$ and $v$ be two words as in the hypothesis. We will show that $\lang(\A_i^{q_u,q_v}) \subseteq L$.
        For ease of notation, we denote $\A_i^{q_u,q_v}$ by $\tilde{\A_i}$.
        Assume $v = w.(\sigma, r)$, and $q_w$ is the unique state that is reached after reading $w$ in $\A_i$. We know that the only new transition in $\tilde{\A_i}$ that was not present in $A_i$ is \blue{$\zeta := q_w \xrightarrow[r]{\sigma}q_u$}.
        
        Let $t\in \lang(\tilde{\A_i})$; we will show that $t \in L$. Consider the run $\rho$ of $t$ in $\tilde{\A_i}$. 
        We prove the result by induction on $k$, that is, the number of occurrences of $\zeta$ in~$\rho$.
        \medskip

        \textbf{Base case ($k = 0$):} 
        In this case, there are three possibilities:
        \item 
        
        (i) $\rho$ is present in $\A_i$ and $q_f$ is an accepting state: then the result holds by outer induction hypothesis~(\ref{IH:L_i_in_L}).

        \item
        (ii) $\rho$ is present in $\Aa_i$, and $q_f$ is not an accepting state: this means, during the merge and subsequent folds, $q_f$ has been merged/folded with an accepting state $q_f'$ of~$\Aa_i$. 
        Since, $q_u$ and $q_v$ have been merged at this step, it must be the case that $q_f'$ is a successor of either $q_u$ or $q_v$. 
        However, if $q_f$ is a successor of $q_v$, since $q_v$ is not red, $q_v$ has only one predecessor $q_w$ and then $\rho$ must visit $\zeta$, that is \blue{$q_w \xrightarrow[r]{\sigma} q_u$}, in $\tilde{\A_i}$ to reach $q_f$, which contradicts the hypothesis that $\rho$ never visits $\zeta$ in $\tilde{\A_i}$.
        Therefore, we know that in $\Aa_i$, there is a path $q_v \xlongrightarrow{y} q_f'$. 
        Then, $\ssp{i}{q_f'} \in L$ since we have shown that $\lang(\Aa_i) \subseteq L$.
        Since $q_v$ is not a red state in $\Aa_i$, $q_f'$ is not red either, and then we know that there is exactly one path from $q_v$ to $q_f'$.
        This implies, $\ssp{i}{q_f'} = \ssp{i}{q_v}.y$. 
        Since $q_f$ has been folded with $q_f'$ and $q_f'$ is a successor of $q_v$, we know that there is a path in $\Aa_i$ from $q_u$ to $q_f$ and moreover, the path must be on the same word $y$, \emph{i.e.} $\Aa_i$ contains the path $q_u \xlongrightarrow{y} q_f$. Then, we have the following:
        \begin{align*}
            & \tail{u.y} = \tail{\ssp{i}{q_f}} ~(\text{using induction step~(\ref{IH:same_state_same_tail})})\\
            \implies & \tail{v.y} = \tail{\ssp{i}{q_f}} ~(\text{since } \tail{u} = \tail{v}) \\
            \implies & \tail{\ssp{i}{q_v}.y} = \tail{\ssp{i}{q_f}} ~(\text{since } \delta_i(q_{in}, v) = q_v = \delta_i(q_{in}, \ssp{i}{q_v})) \\
            \implies & \tail{\ssp{i}{q_f'}} = \tail{\ssp{i}{q_f}} ~(\text{since } \ssp{i}{q_f'} = \ssp{i}{q_v}.y)\\
            \implies & \ssp{i}{q_f} \in L ~(\text{since } \ssp{i}{q_f'} \in L)
        \end{align*}
        Now, since $\rho$ is also a path in $\Aa_i$, we have $\delta_i(q_{in}, t) = q_f$. Then, using induction step~(\ref{IH:same_state_same_tail}) we get $\tail{t} = \tail{\ssp{i}{q_f}}$. Therefore, $t \in L$, since $\ssp{i}{q_f} \in L$.

        \item 
        (iii) $\rho$ is not a path in $\Aa_i$: note that, $q_f$ cannot be red, since otherwise $\rho$ would also be present in $\Aa_i$. Then, 
        $\rho$ can be written in the following from: $q_{in} \xlongrightarrow[\in \A_i]{t_0} q \xlongrightarrow [\notin\A_i] {t_1} q_f$, where $q$ is the last red state in $\rho$. 
        Notice that, since $\rho$ never visits $\zeta$, the first part of the run is in $\A_{i}$.
        Since $q \xlongrightarrow{t_1} q_f$ was not present in $\Aa_i$, and this path does not visit any red states in $\Aa_i$, there must exist a state $q'$ from which there is a path $q' \xlongrightarrow{t_1} q_f$; further, during the merge of $q_u$ and $q_v$, the two states $q$ and $q'$ have been merged/folded.
        Now, similar to the argument in case (ii), we can argue that $q'$ must be a successor of $q_v$, i.e. there exists a path $q_v \xlongrightarrow{y} q'$ in $\Aa_i$.
        Since $q$ and $q'$ have been merged/folded, it must be the case that there is also a path $q_u \xlongrightarrow{y} q$ in $\Aa_i$. We then have the following:
        \begin{align*}
            & \tail{u.y} = \tail{t_0} ~(\text{since } \delta_i(q_{in}, u.y) = q = \delta_i(q_{in}, t_0)) \\
            \implies & \tail{v.y} = \tail{t_0} ~(\text{since } \tail{u} = \tail{v}) \\
            \implies & t = t_0.t_1 \in L ~(\text{since } v.t.t_1 \in \lang(\Aa_i) \subseteq L)
        \end{align*}
        Hence, in all the three cases we get that $t \in L$.
        
        \medskip
        
        \textbf{Induction hypothesis:} We assume that the result holds when the number of occurrences of $\zeta$ is less than $k$.
        \medskip
        
        \textbf{Induction step:} Suppose there are $k$ occurrences of $\zeta$ in $\rho$. Then $\rho$ can be written in the following form:
        \[
        \rho =
        q_{in} \xlongrightarrow[\in \A_i]{t_0} 
        \blue{q_w \xrightarrow[r]{\sigma}
        q_u }\xlongrightarrow[\in \A_i]{t_1} 
        \blue{q_w \xrightarrow[r]{\sigma}
        q_u}
        \cdots
        \blue{q_w \xrightarrow[r]{\sigma}
        q_u} \xlongrightarrow{t_{k}}q_f,
        \]
        where, for every $0 \le i < k$, $t_i \in \RW{\Sigma,K}$, and that every transition in $t_i$ is in $A_i$.
        Consider another word $t'$ and the following run $\rho'$ on $t'$:
        \[
        \rho' =
        q_{in} \xlongrightarrow[]{u} 
        q_u \xlongrightarrow[\in \A_i]{t_1} 
        \blue{q_w \xrightarrow[r]{\sigma}
        q_u}
        \cdots
        \blue{q_w \xrightarrow[r]{\sigma}
        q_u} \xlongrightarrow{t_{k}}q_f
        \]
Since there are $k-1$ occurrences of $\zeta$ in $\rho'$, from inner induction hypothesis, we get that 
$t'  = u.t_1.(\sigma, r). t_2.\ldots (\sigma, r).t_{k}\in L$. Then, using the fact that $\tail{u} = \tail{v}$, and writing $w.(\sigma, r)$ for $v$, we get, $w.(\sigma, r).t_1.(\sigma, r). t_2.\ldots (\sigma, r).t_{k}\in L$.
Finally, since both $t_0$ and $w$ goes to $q_w$ in $\A_i$, using outer induction hypothesis~(\ref{IH:same_state_same_tail}), we get $\tail{t_0} = \tail{w}$. Hence, we conclude $t_0.(\sigma, r).t'\in L$.
    We have proved that $\lang(\A_i^{q_u,q_v}) \subseteq L$. As a consequence, we get $\lang(\Aa_i^{q_u, q_v}) \cap \sminus = \emptyset$.
\end{description}
\end{proof}

Now we can prove Theorem~\ref{th:completeness}.

\begin{proof}[Proof (of Theorem~\ref{th:completeness})]
    The first result follows from Lemma~\ref{lem:finitely-many-tails}. Here we prove the second.
    
    From Lemma~\ref{lem:leap-invariants} (\ref{IH:L_i_in_L}), we get that $\lang(\Aa_n) \subseteq L$. For the reverse direction, let us consider a simple word $w \in L$.
    There exists a word $t \in \sp{L}$ such that $\tail{t} = \tail{w}$. 
    Since $\sp{L} \subseteq \ker{L}$ (Lemma~\ref{lem:sp-subset-nl}), we have that $t \in \ker{L}$. Moreover, since $w \in L$, $t \in L$ as well, then from the construction of the characteristic sample (Definition~\ref{def:char-set-sdera}) we know that $t \in \splus$. 
    Then, again from Lemma~\ref{lem:dfa-language-preserved-leap}, we know that $t \in \lang(\Aa_n)$. 
    If there is a path in $\Aa_n$ on the word $w$, then from Lemma~\ref{lem:leap-invariants} (\ref{IH:same_tail_merge}), we can infer that $\delta_n(q_{in}, w) = \delta_n(q_{in}, t) \in F_n$, and therefore $w \in \lang(\Aa_n)$.
    Otherwise, there is no run of $\Aa_n$ on the word $w$. 
    In that case, $w$ is of the form $w'.(\sigma, r).w''$ such that there is a run $q_{in} \xlongrightarrow{w'} q$ and there is no transition outgoing from $q$ of the form $(q, \sigma, r, *)$.
    From Definition~\ref{def:sp-sdera}, there exists a word $t' \in \sp{L}$ such that $\tail{t'} = \tail{w'}$. Now, since $w = w'.(\sigma, r).w'' \in L$, we get that $t'.(\sigma,g).w'' \in L$. This implies, $t'. (\sigma, r) \in \ker{L}$.
    From Definition~\ref{def:char-set-sdera}, we know that $\splus$ contains a simple word $t = t'.(\sigma, r).t''$. Therefore, there must exist a run of $\Aa_n$ on $t$. 
    Since $\Aa_n$ has runs on both $w'$ and $t'$ and, moreover, since $\tail{t'} = \tail{w'}$, from Lemma~\ref{lem:leap-invariants}(5) we know that $\delta_n(q_{in}, w') = \delta_n(q_{in}, t')$.
    This contradicts the fact that there are no outgoing transitions $(q,\sigma,g,*)$.
\end{proof}

The characteristic sets that we construct (as stated in Theorem~\ref{th:completeness}) consist of only region words, and the size of these sets can be exponential in the worst case w.r.t. the size of a minimal $K$-DERA recognizing $L$, as justified in Corollary~\ref{cor:no-poly-ch-set}.
Due to the large number of samples present in these characteristic sets, it is not practical to use \leap\ on such sets, but it is still interesting (from a theoretical point of view) to prove that \leap\  indeed terminates and returns the correct language. 
On the other hand, notice that Corollary~\ref{cor:no-poly-ch-set} does not rule out the possibility of the existence of characteristic sets of polynomial size  containing general symbolic timed words. We plan on exploring this possibility in the future. Indeed, this possibility is substantiated by the empirical evidence provided in Section~\ref{sec:experiments}, where we will demonstrate that, in practice, one can often construct reasonably small sample sets for ERA-recognizable languages, containing general symbolic words, that are correct for \leap.

\section{Empirical evaluation}
\label{sec:experiments}

\begin{table}[t]
    \caption{$|S| = |\splus| + |\sminus|$; 
   $\#$queries  
    for $\mathsf{tLsep}$ is the number of membership and inclusion queries; for $\mathsf{LearnTA}$, this is the number of membership and equivalence queries. Times are reported in seconds; `TO'  denotes a timeout of 900s.} 
    \centering
    \begin{tabular}{|c||c|c|c||c|c|c||c|c|c||c|c|c||}
         \hline
     & \multicolumn{3}{c||}{\leap (zones)} & \multicolumn{3}{c||}{\leap (regions)\footnotemark{}} & \multicolumn{3}{c||}{$\mathsf{tLsep}$~\cite{ATVA2024}} & \multicolumn{3}{c||}{$\mathsf{LearnTA}$~\cite{Waga23}} \\
    \hline
     Model & $|S|$ & $|Q|$ & time & $|S|$ & $|Q|$ & time & $\#$queries & $|Q|$ & time & $\#$queries & $|Q|$ & time \\
    \hline
    \hline
    Figure~\ref{fig:motivation-2} & 8 & 3 & 0.10 & 75 & 3 & 87.29 & 339 & 6 & 6.43 & 131 & 3 & 0.01\\
    \hline
    Figure~\ref{fig:dera}  & 7 & 3 & 0.13 & 56 & 2 & 49.21 & 106 & 2 & 1.66 & 268 & 3& 0.01\\
    \hline
     Figure~\ref{fig:fig6} & 8 & 3 & 0.16 & 24 & 3 & 4.53 & 31 & 3 & 0.22 & 45 & 3 & 0.01\\
    \hline
    \hline
     Unbalanced-1 & 16 & 4 & 0.13 & -- & -- & TO & 438 & 9 & 10.64 & 2400 & 8 & 0.04  \\
    \hline
     Unbalanced-2 & 16 & 4 & 0.12 & -- & -- & TO & 1122 & 14 & 47.12 & 13671 & 12 & 0.17 \\
     \hline
     Unbalanced-3 & 16 & 4 & 0.12 & -- & -- & TO & 2124 & 18 & 343.34 & 46220 & 16 & 1.120 \\
     \hline
     PC & 176 & 19 & 147.85 & -- & -- & TO & -- & -- & TO & 1613401 & 38 & 607.45 \\
     \hline
    \hline
    $L_2$ & 10 & 4 & 0.20 & -- & -- & TO & 246 & 6 & 1.51 & 513 & 5 & 0.01 \\
    \hline
    $L_4$ & 16 & 6 & 0.90 & -- & -- & TO & 778 & 18 & 17.60 & 3739 & 17 & 0.15\\
    \hline
     $L_{8}$ & 28 & 10 & 12.39 & --  & -- & TO & -- & -- & TO & 141635 & 257 &  100.57 \\
     \hline
    \end{tabular}
    \label{tab:experiments}
\end{table}

We have implemented a prototype of $\leap$ in \textsc{Python}. The steps shown in Figures~\ref{fig:motivation-2} and~\ref{fig:running-ex-illustration} are outputs from this implementation. Here, we empirically compare \leap's performance with the implementations of other learning algorithms applicable to the class of ERA. 
\footnotetext{the samples for \leap(region) are constructed from the samples of \leap(zones) by \emph{systematically} breaking each symbolic timed word into all its included region words}

Existing passive learning approaches differ significantly from ours. Works such as \cite{VWW10,PMM17,CLRT22} learn from plain timed words rather than {\em symbolic timed words} as we do; \cite{PMM17,CLRT22} use only positive data, and \cite{VWW10,CLRT22} infer automata from a weaker class. 
A merging based algorithm was proposed in~\cite{GJP06}, however, it requires examples that meet certain completeness criteria, which our method does not,  and moreover,  implementation of their algorithm is not available.
These differences make a direct comparison of these algorithms with $\leap$ inappropriate. 
On the other hand, as remarked in the introduction of this article, $\leap$ can be used interactively, with the requirement engineer providing additional symbolic timed words representing unmet requirements, which $\leap$ then incorporates. This resembles an active learning approach where the user acts as the teacher and $\leap$ as the learner. Recent works propose active learning frameworks for ERA~\cite{ATVA2024} and Deterministic Timed Automata (DTA)~\cite{DTA-hscc,Waga23}. To the best of our knowledge, the implementation in~\cite{DTA-hscc} is not available.
Therefore, we compare our algorithm with following three methods: (1) a version of \leap\ restricted to region words (reported in the column ``\leap(regions)'' in Table~\ref{tab:experiments}), highlighting the advantage of working with general symbolic (zone) words (the column ``\leap(zones)'' in Table~\ref{tab:experiments}); (2)~{\sf tLsep}~\cite{ATVA2024}, a recent active learning tool for ERA; and (3)~{\sf LearnTA}~\cite{Waga23}, an active learning tool for deterministic TA. We use the number of samples (all the reported algorithms operate on symbolic words) as the metric for comparison: for active learning, it is the number of queries, while for passive learning, it is $|\splus| + |\sminus|$.

\begin{figure}[t]
    \centering
    \begin{tikzpicture}[scale = 1]
        \everymath{\scriptstyle}
        \begin{scope}[every node/.style={circle, draw, inner sep=2pt,
            minimum size = 3mm, outer sep=3pt, thick}] 
            \node [double] (0) at (0,0) {$q_0$}; 
            \node [] (1) at (2.1,0) {$q_1$}; 
            \node [double] (2) at (5.2,0) {$q_2$}; 
        \end{scope}
        \begin{scope}[->, thick]
            \draw (-0.7,0) to (0); 
            \draw [rounded corners] (0) to (-0.3,1) to (0.3, 1) to (0);
            \draw (0) to (1);
            \draw [rounded corners] (1) to (1.8,1) to (2.4, 1) to (1);
            \draw [bend left] (1) to (2);
            \draw (1) to (2);
            \draw [bend left] (2) to (1);
        \end{scope}
        \node at (0,1.2) {$a,~x_a > 0$};
        \node at (1.1,0.2) {$a,~x_a = 0$};
        \node at (2,1.2) {$a,~0 < x_a < 1$};
        \node at (3.5,0.2) {$a,~x_a = 0$};
        \node at (3.5,0.8) {$a,~x_a \geq 1$};
        \node at (3.5,-0.8) {$a,~x_a \geq 0$};
    \end{tikzpicture}
    \caption{Example ERA}
    \label{fig:fig6}
\end{figure}

Table~\ref{tab:experiments} summarizes the experimental results. The first two rows in the table represent timed languages recognized by the ERA in Figures~\ref{fig:motivation-2} and~\ref{fig:dera}, respectively. The third model corresponds to the language of the automaton in Figure~\ref{fig:fig6}, while \textsf{Unbalanced-$k$} and \textsf{PC} are all DTA models (ERA, after clock renaming)  used as benchmarks for \textsf{LearnTA}. 
The last three examples are instances of a family of languages $L_n$ that showcases the compactness of nondeterministic ERA over equivalent DERA, which
is formally defined as follows. Let $\Sigma = \{a, b\}$; then  $L_n$ is a timed language over $\Sigma_{\#} := \Sigma \cup \{{\#}\}$ that accepts timed words satisfying the conditions: (i) they begin with the symbol $\#$, (ii) $\#$ does not reappear, and (iii) the $n$-th last letter is $a$ which occurrs exactly 1 time unit after reading~$\#$. 

The experiments
were performed in a machine with 64GB of RAM running Fedora Linux~41. 
All the times that are reported in the table are reported in seconds, computed using the command-line tool \verb|time|. 
The tool $\mathsf{LearnTA}$ takes significantly less running time compared to the other tools, possibly because the former is written in \textsc{C++} while the remaining ones are written in \textsc{Python}. 

Note that, strangely, even though $\mathsf{tLsep}$~\cite{ATVA2024} promises to compute ERA with the minimum number of control states, for some examples in Table~\ref{tab:experiments} we see that this is not the case. This could perhaps be because, in the implementation of $\mathsf{tLsep}$~\cite{ATVA2024} the authors implement a heuristic which is more efficient but does not guarantee minimality anymore, or possibly $\mathsf{LearnTA}$ can learn more compact models since they can reset the clocks to arbitrary values and also can output automata from a larger class. 

\medskip
We now illustrate the main advantages of \leap\ as witnessed in our experiments. 

\subparagraph*{1. Small number of samples.}
As we can observe in Table~\ref{tab:experiments}, ``\leap(zones)'' requires significantly fewer symbolic words as input compared to the number of queries needed by $\mathsf{tLsep}$ and $\mathsf{LearnTA}$. This difference arises because the symbolic zone  words used as input for $\leap$ are based on specific intuitions from the informal requirements, that a requirements engineer (RE) can provide, delivering rich information about the target language. In contrast, active learning algorithms like $\mathsf{tLsep}$ and $\mathsf{LearnTA}$ must derive all aspects of the target language through {\em systematic} membership and equivalence queries.
While we do not claim that $\leap$ is universally superior to active learning algorithms, it is better suited for helping an RE craft an ERA from scenarios, as it avoids the need for the RE to respond to numerous membership queries, as required in active learning setups.

\subparagraph*{2. Insensitivity to the maximal constant.} 
Another advantage of being able to handle symbolic timed words directly is demonstrated by the two \textsf{Unbalanced} examples -- the two models here are identical, except the maximum constant being $1$ and $3$, respectively. Since symbolic timed words are not dependent on this maximum constant, the sample size required by \leap\ remains the same, while the performances of the active learning algorithms take a considerable hit, due to their sensitivity to this constant.

\subparagraph*{3. Advantages of zone words.} 
In Section~\ref{sec:leap-completeness-sdera}, we showed that for every ERA-recognizable language, there exists characteristic sets 
for which \leap\ returns the correct automaton, however, these sets contain only region words and can be prohibitively large. 
Even though handling region words is computationally easier than handling general symbolic timed words (Theorem~\ref{thm:complexity}), the smaller values in the column ``\leap(zones)'' compared to the ones in ``\leap(regions)'' showcase the advantages of working with symbolic timed words. Splitting them into region words can very quickly result in an enormous blowup in the number of samples, as evidenced by the models {\sf Unbalanced-}$k$ and $L_n$'s in the table. 
Symbolic timed words can usually represent large sets of region words at once.
The example \textsf{PC} is seemingly the largest (in terms of number of states) practical benchmark as mentioned in~\cite{Waga23}. Due to the conciseness of zone words, \leap\ manages to learn the automaton for this example, showing its ability to scale.

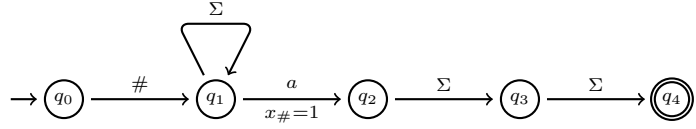
\begin{figure}[t]
    \centering
        \begin{tikzpicture}
        \everymath{\scriptstyle}
        \begin{scope}[every node/.style={circle, draw, inner sep=2pt,
            minimum size = 3mm, outer sep=3pt, thick}] 
            \node [] (0) at (0,0) {$q_0$}; 
            \node [] (1) at (2,0) {$q_1$};
            \node [] (2) at (4,0) {$q_2$};
            \node [] (3) at (6,0) {$q_3$};
            \node [double] (4) at (8,0) {$q_4$};
        \end{scope}
        \begin{scope}[->, thick]
            \draw (-0.7,0) to (0); 
            \draw (0) to (1);
            \draw [rounded corners] (1) to (1.5,1) to (2.5, 1) to (1);
            \draw (1) to (2);
            \draw (2) to (3);
            \draw (3) to (4);
        \end{scope}
        \node at (1,0.2) {$\#$};
        \node at (3,0.2) {$a$};
        \node at (3,-0.2) {$ x_{\#} = 1$};
        \node at (2,1.2) {$\Sigma$};
        \node at (5,0.2) {$\Sigma$};
        \node at (7,0.2) {$\Sigma$};
        
    \end{tikzpicture}
    \caption{An ERA recognizing the language $L_3$.}
    \label{fig:L3}
\end{figure}

\subparagraph*{4. Advantages of maintaining nondeterminism.} 
As noted earlier, each step in \leap\ produces a possibly non-deterministic ERA, allowing it to learn smaller models for a given language compared to existing algorithms in the literature, which target deterministic models~\cite{GJL10,Waga23,ATVA2024}. 
With standard techniques, one can show that any DERA recognizing $L_n$ defined earlier requires at least $O(2^n)$ states (also witnessed in the outputs of \textsf{tLSep} and \textsf{LearnTA}), while there exists non-deterministic ERA $A_n$ recognizing $L_n$ 
with only $n+2$ states (e.g., $A_3$ in Figure~\ref{fig:L3} recognizes $L_3$). 
To learn $L_n$, we construct a sample set $S_n$ of size linear in $n$ as defined below, containing symbolic words, which when provided to \leap, yields $A_n$. 
    \begin{align*}
    S^n_+ = \{&(\#, \top).(a,\top).(a, x_{\#} \le 1).(a, \top)^{n-1}, (\#, \top).(b, \top).(a,x_{\#} \le 1).(b, \top)^{n-1}\}; \\
    S^n_- = \{&\varepsilon, (\#,\top)(a,\top)^{n-1}, (\#,\top)(b,\top)^{n-1}, (a,\top)(b,\top)^{n-1},\\ &\{(a,\top)^i\mid 1 \le i \le {n+1}\},\\ &\{(b,\top)^i\mid 1 \le i \le {n-1}\},\\ &\{(\#,\top)(a,\top)(b,\top)^i \mid 1 \le i \le {n-2}\}\}.
    \end{align*}
Table~\ref{tab:experiments} compares the size of $S_n$ with the number of queries required by \textsf{tLSep} and \textsf{LearnTA} for various values of $n$. 

\medskip
Note that, the values reported for the tool $\mathsf{LearnTA}$ on the benchmarks taken from~\cite{Waga23} are different from the ones reported in~\cite{Waga23}. This is because, these are modelled as DTAs in~\cite{Waga23}, whereas here we model these as ERAs. The ERA models are larger in size and hence $\mathsf{LearnTA}$ takes more time.

\section{Discussion and future works}
\label{sec:conclusion}

In this work, we introduced a merging-based passive learning algorithm for ERA-definable languages. We proved that determining if a merge is possible is an $\np$-complete problem, and we have provided an SMT-based solution for this problem. As shown in our experiments, \leap\ successfully learns ERA from small sets of symbolic timed words as positive and negative samples. We established the completeness of \leap\ using a characteristic sample set, that is guaranteed to exist for all ERA-recognizable languages.

As demonstrated, \leap\ is well-suited to scenarios where a user guides the learning process, making symbolic timed words a natural input format. However, \leap\ can also be readily applied to timed words (obtained from logs for instance). With a known maximal constant $K$ and a fixed granularity for atomic constraints, a timed word can be converted to a region word by constructing its corresponding clocked word and replacing each valuation with its unique $K$-simple constraint. \leap\ can then be applied to these transformed samples.

\smallskip
\noindent\textbf{In future} we aim to define characteristic sets that include symbolic words. As noted in Section~\ref{sec:leap-completeness-sdera}, polynomial size characteristic sets containing symbolic words may still exist for ERA recognizable languages.
Secondly, we would like to explore the possibility of generating the examples using LLMs instead of handcrafting.
This is because, handcrafting many symbolic timed words (for instance, $176$ in the PC example), some being relatively long, may not be practical for a requirement engineer. However, this opens up an interesting potential application of Large Language Models (LLM), where, instead of handcrafting the examples, the requirement engineer may specify the requirement in natural language and ask the LLM to produce examples and only verify the validity of each of the examples before providing them to \leap. 
While asking to turn a natural language specification into a fully-fledged ERA is out of reach of LLMs, it appears that generating positive and negative examples in the form of symbolic timed words is feasible even with off-the-shelf LLMs. 
Further, an incremental approach to learning ERA, where the sample is presented one word at a time instead of all at once in the beginning, similar to the work in~\cite{Dupont96} for regular languages, would be an interesting direction to explore. 

In~\cite{DBLP:conf/tacas/BalachanderFR23}, 
a framework combines reactive synthesis from LTL specifications with example-based learning to produce a Mealy machine that satisfies an LTL formula and aligns with specified execution traces, independent of environmental behavior. Our merging algorithm for timed languages suggests a natural extension of this approach to real-time settings, where the  input can be a real-time formula (say, in MITL~\cite{AFH96}), along with symbolic words as \emph{hints}  that our algorithm could process.

\bibliographystyle{plain}
\bibliography{main}

\end{document}